\documentclass[journal,onecolumn]{IEEEtran}
\usepackage[T1]{fontenc}
\usepackage{amsmath,amsfonts,amsthm, mathabx}
\usepackage{algorithmic,lipsum}
\usepackage{array}
\usepackage[caption=false,font=normalsize,labelfont=sf,textfont=sf]{subfig}
\usepackage{textcomp,tcolorbox}
\usepackage{mathrsfs}
\usepackage{stfloats}
\usepackage{hyperref}       
\hypersetup{
  colorlinks   = true, 
  urlcolor     = blue, 
  linkcolor    = blue, 
  citecolor   = blue 
}
\usepackage{url}
\usepackage{float}
\usepackage{verbatim}
\usepackage{graphicx}
\def\BibTeX{{\rm B\kern-.05em{\sc i\kern-.025em b}\kern-.08em
    T\kern-.1667em\lower.7ex\hbox{E}\kern-.125emX}}
\usepackage{bbm}
\usepackage{physics}
\newtheorem{theorem}{Theorem}[section]

\newtheorem{lemma}[theorem]{Lemma}
\newtheorem{remark}[theorem]{Remark}
\newtheorem{definition}[theorem]{Definition}
\newtheorem{prop}[theorem]{Proposition}

\newtheorem{conjecture}[theorem]{Conjecture}
\usepackage{tikz-cd}
\usepackage{tikz}
\usetikzlibrary{positioning,arrows.meta}
\usetikzlibrary{decorations.pathreplacing}
\usepackage[
  left=0.5in,
  right=0.5in,
  top=0.5in,
  bottom=0.5in,
]{geometry}
\usetikzlibrary{arrows, automata}
\usetikzlibrary{shapes.geometric, arrows}
\tikzstyle{block} = [rectangle, rounded corners, minimum width=2.5cm, minimum height=1cm, text centered, draw=black, fill=blue!20]
\tikzstyle{arrow} = [thick,->,>=stealth]
\tikzstyle{shield} = [rectangle, rounded corners, minimum width=2.5cm, minimum height=1cm, text centered, draw=black, fill=red!20]

\usepackage{color,hyperref}
\usepackage{physics}

\newcommand{\wt}{\widetilde}

\newcommand{\mb}{\mathbb}
\newcommand{\mc}{\mathcal}

\renewcommand{\ket}[1]{\lvert #1\rangle}
\renewcommand{\bra}[1]{\langle #1\rvert}
\renewcommand{\ketbra}[2]{\ket{#1}\!\bra{#2}}

\definecolor{cool_green}{rgb}{0.0, 0.5, 0.0}

\title{Single-letter one-way distillable entanglement for non-degradable states}
\author{%
Rabsan Galib Ahmed$^{1,2}$, Graeme Smith$^{1,2}$ and Peixue Wu$^{1,2}$ \\[1ex]
\small{
\begin{tabular}{l}
$^{1}$\textit{Institute for Quantum Computing, University of Waterloo,}
\textit{200 University Avenue West, Waterloo, ON N2L 3G1, Canada} \\[1ex]
$^{2}$\textit{Department of Applied Mathematics, University of Waterloo,}
\textit{200 University Avenue West, Waterloo, ON N2L 3G1, Canada} \\[1ex]
\end{tabular}
}
}
\begin{document}
\maketitle
\begin{abstract}
The one-way distillable entanglement is a central operational measure of bipartite entanglement, quantifying the optimal rate at which maximally entangled pairs can be extracted by one-way LOCC. Despite its importance, it is notoriously hard to compute, since it is defined by a regularized optimization over many copies and adaptive one-way protocols. At present, single-letter formulas are only known for (conjugate) degradable and PPT states. More generally, it has remained unclear when one-way distillable entanglement can still be additive beyond degradability and PPT settings, and how such additivity relates to additivity questions of quantum capacity of channels.

In this paper, we address this gap by identifying three explicit families of non-degradable and non-PPT states whose one-way distillable entanglement is nevertheless single-letter. First, we introduce two weakened degradability-type conditions—regularized less-noisy and informationally degradable—and prove that each guarantees additivity and hence a single-letter formula. Second, we show a stability result for orthogonally flagged mixtures: when one component has orthogonal support on Alice’s system and zero one-way distillable entanglement, the mixture remains single-letter, even though degradability is typically lost under such mixing. Finally, we propose a generalized spin-alignment principle for entropy minimization in tensor-product settings, which we establish in several key cases, including a complete R\'enyi-2 result. As an application, we obtain additivity results for generalized direct-sum channels and their corresponding Choi states.
\end{abstract}
\tableofcontents








\section{Introduction}

Distilling maximally entangled states from many copies of a noisy bipartite state
via local operations and classical communication (LOCC)
\cite{Bennett1996Concentrating,Bennett1996Purification,Bennett96QECC_distillation}
is a foundational problem in quantum information science, enabling key protocols such as
dense coding and teleportation \cite{Bennett1992DenseCoding,Bennett1993Teleportation}.
When the allowed classical communication is restricted to be one-way (say $A\to B$),
the optimal asymptotic rate at which ebits can be extracted is quantified by the
\emph{one-way distillable entanglement} $D_\to(\rho_{AB})$.
Beyond its operational meaning as an entanglement measure, $D_\to$ is tightly connected to
communication problems: through channel--state duality, additivity and single-letter
phenomena for one-way distillation mirror (but are not identical to) additivity questions
for coherent information and quantum capacity.

A major obstruction is that $D_\to$ is generally \emph{non-additive}
\cite{Horodecki2009RMP,Shor_03}, and therefore is defined through a regularization
\cite{Devetak_2005}:
\begin{align*}
    D_\to(\rho_{AB}) = \lim_{n\to \infty} \frac{1}{n} D_\to^{(1)}(\rho_{AB}^{\otimes n}).
\end{align*}
Determining when such regularizations collapse to single-letter
formulas is a central theme in quantum information theory.
For (conjugate) degradable and anti-degradable states, one-way distillable entanglement
admits a single-letter expression given by coherent information \cite{Leditzky_2018}.
Likewise, PPT states have vanishing two-way (and hence one-way) distillable entanglement
\cite{Horodecki_98}.

In contrast, explicit families of states that are \emph{non-degradable}, non-PPT, and yet satisfy
\[
D_\to(\rho_{AB}) = D^{(1)}_\to(\rho_{AB})
\]
have remained unexplored. Motivated by recent progress on additivity for non-degradable channels in quantum capacity theory
\cite{Leditzky_2023,Smith_2025}, we fill this gap by introducing three mechanisms that produce such states:
\begin{enumerate}
    \item[(i)] \emph{states satisfying weaker notions of degradability},
\item[(ii)] \emph{orthogonally mixtures with a ``useless'' component}, and
\item[(iii)] \emph{states governed by a spin-alignment entropy minimization principle}.
\end{enumerate}
\medskip
\noindent\textbf{A channel-capacity viewpoint and a subtle mismatch.}
Given a bipartite state $\rho_{AB}$, a natural perspective suggested by the Choi correspondence is to compare
additivity phenomena for one-way distillation with those for one-shot channel capacity.
Concretely, let $\Phi_\rho^{A\to B}$ be the completely positive map induced by $\rho_{AB}$ via the
Choi--Jamio\l{}kowski isomorphism (equivalently, $\rho_{AB}$ is proportional to the Choi operator of $\Phi_\rho$).
One may ask whether additivity of the maximal coherent information $\mathcal{Q}^{(1)}(\Phi_\rho)$
should be equivalent to additivity of $D_\to^{(1)}(\rho_{AB})$.
While these two questions coincide in the (conjugate) degradable setting, in general they are
\emph{a priori different}: the optimizations defining the two quantities impose different operational constraints.

Recall \cite{Devetak_2005} that
\begin{align*}
D^{(1)}_\to(\rho_{AB})
= \max \Big\{ I(A'\rangle BM)_{\omega_{A'BM}}:\;
\omega_{A'BM} = \sum_m (K_m \otimes \mathbb{I}_B)\rho_{AB}(K_m^\dagger\otimes \mathbb{I}_B)\otimes |m\rangle\langle m|,
\ \sum_m K_m^\dagger K_m = \mathbb{I}_A \Big\}.
\end{align*}
A natural relaxation is to \emph{drop the completeness constraint} $\sum_m K_m^\dagger K_m=\mathbb{I}_A$,
i.e., allow trace non-preserving operations and renormalize.
For any $K:A\to A'$, define the postselected state
\[
\omega_{A'B}(K)
= \frac{(K\otimes \mathbb{I}_B)\rho_{AB}(K\otimes \mathbb{I}_B)^\dagger}
{\mathrm{Tr}\big[(K^\dagger K\otimes \mathbb{I}_B)\rho_{AB}\big]}.
\]
Then one obtains the simplified expression
\begin{equation}\label{eq:hatD-singleK}
\widehat D^{(1)}_\to(\rho_{AB})
= \max_{K}\; I(A'\rangle B)_{\omega_{A'B}(K)} ,
\end{equation}
which upper bounds $D^{(1)}_\to(\rho_{AB})$ by construction.

This relaxation admits an exact channel interpretation.
Let $\rho_{AB}^{\mathcal{N}}$ denote the normalized Choi state of a channel $\mathcal{N}^{A\to B}$.
Using the Choi--Jamio\l{}kowski isomorphism, $\widehat D^{(1)}_\to$ becomes precisely the one-shot
quantum capacity:
\[
\mathcal{Q}^{(1)}(\mathcal{N})
= \widehat D^{(1)}_\to\!\big(\rho_{AB}^{\mathcal{N}}\big),
\]
where $\mathcal{Q}^{(1)}(\mathcal{N})$ is the maximal coherent information of $\mathcal{N}$.

Consequently, additivity of $\widehat D^{(1)}_\to$ is equivalent (via the Choi correspondence)
to additivity of $\mathcal{Q}^{(1)}$ for the associated channels.
In contrast, additivity for $D^{(1)}_\to$ is intrinsically more delicate, since the optimization is restricted to
\emph{complete} one-way instruments $\sum_m K_m^\dagger K_m=\mathbb{I}_A$ rather than arbitrary
postselected filters.
Our three mechanisms provide structured settings in which this completeness constraint can still be controlled,
yielding explicit non-degradable, non-PPT families with single-letter one-way distillable entanglement.


We present three structural mechanisms that yield explicit families of \emph{non-degradable}, non-PPT states
with single-letter one-way distillable entanglement.  For each mechanism we first state the guiding principle,
and then give a representative example.

\paragraph{Mechanism 1: additivity from weaker degradability (information dominance)}

Let $\rho_{AB}$ have a purification $\ket{\phi}_{ABE}$.  Standard degradability requires that $E$ can be simulated
from $B$ by a channel; we instead impose only an \emph{information ordering} saying that, under all relevant
pre-processings on Alice, Bob is at least as informative as the purifying system.
Two convenient formulations are:
\begin{itemize}
\item \textbf{(Regularized less noisy.)} For every $n\ge 1$ and every quantum instrument
$\mathcal{T}^{A^n\to A'M}$, letting $\omega_{A'B^nE^n}=\mathcal{T}(\phi_{ABE}^{\otimes n})$,
\begin{equation}\label{eq:less-noisy-level-n-abs}
I(M;B^n)_\omega \ge I(M;E^n)_\omega .
\end{equation}
\item \textbf{(Informationally degradable.)} For every channel $\mathcal{N}^{A\to A'}$ and
$\omega_{A'BE}=\mathcal{N}(\phi_{ABE})$,
\begin{equation}\label{eq:ID-abs}
I(A';B)_\omega \ge I(A';E)_\omega .
\end{equation}
\end{itemize}
These conditions are strictly weaker than degradability but are still strong enough to force
single-letter behavior, by the same telescoping/chain-rule strategy that underlies many additivity
proofs in capacity theory (cf.\ \cite{Smith_2025}).

\smallskip
\noindent\textbf{Consequence (informal).}
If \eqref{eq:less-noisy-level-n-abs} holds for all $n$, then
\[
D_\to(\rho_{AB}) = D_\to^{(1)}(\rho_{AB}) = I(A\rangle B)_\rho,
\]
and $D_\to^{(1)}$ is additive for many copies of such states, i.e., it exhibits weak additivity under tensor products.
Furthermore informationally degradable states satisfy additivity for $D_\to^{(1)}$ under tensor products i.e., it exhibits strong additivity under tensor products.

\medskip
\emph{Example (flagged mixtures of amplitude damping channels).}
Let ${\rm AD}_\gamma$ be the qubit amplitude damping channel with Kraus operators
\[
E_0 = |0\rangle\!\langle 0| + \sqrt{1-\gamma}\,|1\rangle\!\langle 1|,
\qquad
E_1 = \sqrt{\gamma}\,|0\rangle\!\langle 1| ,
\qquad \gamma\in[0,1].
\]
Fix parameters $\gamma_0,\gamma_1\in[0,1]$ and a mixing probability $p\in[0,1]$, and define the flagged channel
\[
\mathcal{N}^{A\to BF}(\rho)
:= p\,{\rm AD}_{\gamma_0}(\rho)\otimes |0\rangle\!\langle 0|_F
 + (1-p)\,{\rm AD}_{\gamma_1}(\rho)\otimes |1\rangle\!\langle 1|_F ,
\]
where the classical flag $F$ is available to Bob.  Its (normalized) Choi state
$\rho_{RBF}^{\mathcal{N}} := (\mathrm{id}_R\otimes \mathcal{N})(|\Phi^+\rangle\!\langle\Phi^+|)$,
with $|\Phi^+\rangle = (|00\rangle+|11\rangle)/\sqrt{2}$, is
\begin{equation}\label{eq:choi-flagged-AD}
\rho_{RBF}^{\mathcal{N}}
= p\,\rho_{RB}^{(\gamma_0)}\otimes |0\rangle\!\langle 0|_F
 + (1-p)\,\rho_{RB}^{(\gamma_1)}\otimes |1\rangle\!\langle 1|_F,
\end{equation}
where, in the computational basis $\{|00\rangle,|01\rangle,|10\rangle,|11\rangle\}$ of $RB$,
\begin{equation}\label{eq:choi-AD-explicit}
\rho_{RB}^{(\gamma)}
=
\frac{1}{2}\begin{pmatrix}
1 & 0 & 0 & \sqrt{1-\gamma}\\
0 & 0 & 0 & 0\\
0 & 0 & \gamma & 0\\
\sqrt{1-\gamma} & 0 & 0 & 1-\gamma
\end{pmatrix}.
\end{equation}
This flagged family for some parameter regions provides explicit \emph{non-degradable} Choi states for which the weaker
information-dominance property holds and hence $D_\to(\rho_{RBF}^{\mathcal{N}})$ is single-letter~\cite{Smith_2025}.

\paragraph{Mechanism 2: orthogonal flags with a ``useless'' component}

Suppose Alice locally knows which component of a mixture she holds, i.e.\ the components have orthogonal support
on $A$.  Then one-way distillation can be performed conditionally on that classical information.
If, moreover, one component is ``useless'' for one-way distillation (e.g.\ anti-degradable / separable, hence
zero one-way distillable entanglement), the mixture inherits a single-letter formula from the ``useful''
component.  Operationally, orthogonality on $A$ turns the problem into a direct-sum decomposition.

\smallskip
\noindent\textbf{Prototype statement (informal).}
Let $\rho_0,\rho_1$ be bipartite states with $\rho_0^A\perp \rho_1^A$ and assume the product additivity bound
$D_\to^{(1)}(\bigotimes_{i=1}^n \rho_{w_i})\le \sum_{i=1}^n D_\to^{(1)}(\rho_{w_i})$ for all words $w^n$.
Then for any $p\in[0,1]$,
\[
D_\to\big(p\rho_0+(1-p)\rho_1\big)
= D_\to^{(1)}\big(p\rho_0+(1-p)\rho_1\big)
= p\,D_\to(\rho_0)+(1-p)\,D_\to(\rho_1).
\]

\medskip
\emph{Example (a $3\times 3$ state with a separable ``junk'' block).}
Consider, for $s\in[0,1]$, the state on $\mathbb{C}^3\otimes\mathbb{C}^3$
\begin{align*}
\rho_{AB}(s)
&:= \frac{1}{3}\Big(
|0\rangle\!\langle0|_A\otimes\big[s\,|0\rangle\!\langle0|_B+(1-s)\,|2\rangle\!\langle2|_B\big]
+|0\rangle\!\langle1|_A\otimes \sqrt{s}\,|0\rangle\!\langle2|_B\\
&\hspace{2.2cm}
+|1\rangle\!\langle0|_A\otimes \sqrt{s}\,|2\rangle\!\langle0|_B
+|1\rangle\!\langle1|_A\otimes |2\rangle\!\langle2|_B\Big)
+\frac{1}{3}|22\rangle\!\langle22|_{AB}.
\end{align*}
Here the last term $\frac{1}{3}|22\rangle\!\langle22|$ is separable and supported on an $A$-subspace orthogonal to
the support of the first block.  The remaining ``useful'' component is a Choi-type amplitude-damping structure on
${\rm span}\{|0\rangle,|1\rangle\}_A$.
By the above principle, $\rho_{AB}(s)$ has single-letter one-way distillable entanglement, and one obtains
\[
D_\to(\rho_{AB}(s)) = D_\to^{(1)}(\rho_{AB}(s))
= \frac{2}{3}\max\!\left\{ h\!\left(\frac{s}{2}\right) - h\!\left(\frac{1+s}{2}\right),\,0 \right\}.
\]

\paragraph{Mechanism 3: spin alignment and generalized direct-sum structure}

Certain block-structured channels/states reduce multi-copy optimizations (entropy minimization or coherent
information maximization) to a classical mixture over ``which block is used''~\cite{wu2025}.
The key step is a \emph{spin-alignment} entropy minimization rule: when outputs are mixtures of the form
$\rho\otimes\sigma_1$ and $\sigma_0\otimes\rho$, entropy is minimized by aligning the free input $\rho$
with maximal-eigenvalue directions of the fixed states $\sigma_0,\sigma_1$.
This forces optimal $n$-copy inputs to be product-aligned across sites, turning a priori noncommutative
problems into tractable classical ones.

\smallskip
\noindent\textbf{Spin alignment (informal).}
Fix $\sigma_0\in\mathcal{D}(B_0)$ and $\sigma_1\in\mathcal{D}(B_1)$ and define
$\mathcal{N}_0(\rho)=\rho\otimes\sigma_1$, $\mathcal{N}_1(\rho)=\sigma_0\otimes\rho$.
Given a distribution $\{p_{\vec x}\}$ on $\{0,1\}^n$, consider
\[
\min_{\{\rho_{\vec x}\}} S\!\left(\sum_{\vec x} p_{\vec x}\,\mathcal{N}_{\vec x}(\rho_{\vec x})\right).
\]
The conjectured (and partially proved) alignment rule says the optimum is attained by product inputs
$\rho_{\vec x}=\bigotimes_i \tau^{(i)}_{x_i}$ with each $\tau^{(i)}_{0}$ (resp.\ $\tau^{(i)}_{1}$)
a projector onto a maximal-eigenspace of $\sigma_0$ (resp.\ $\sigma_1$).
We establish the $n=1$ von Neumann case, and the full $n$-copy statement for R\'enyi-$2$ entropy.

\medskip
\emph{Example (a generalized direct-sum state with single-letter $D_\to$).}
Let $A,B\simeq \mathbb{C}^{d_0+d_1}$ and define
\[
\rho_{AB}
=
\frac{1}{2d_0 d_1}
\sum_{i=0}^{d_0-1}
\sum_{j=0}^{d_1-1}
\Big(
|i\rangle|j\rangle
+
|j+d_0\rangle|i+d_1\rangle
\Big)
\Big(
\langle i|\langle j|
+
\langle j+d_0|\langle i+d_1|
\Big).
\]
This state is a canonical ``two-block coherent coupling'' (a Choi-type generalized direct-sum structure).
Using the alignment principle to control the relevant entropy minimizations, we show that
\[
D_\to(\rho_{AB}) = D_\to^{(1)}(\rho_{AB}),
\]
providing an explicit non-degradable, non-PPT family where one-way distillation is single-letter.

\medskip
\noindent
Together, these three mechanisms---information dominance without full degradability,
orthogonal flags with a useless component, and spin-alignment-driven block reduction---yield broad,
explicit sources of additivity and single-letter behavior for $D_\to$ beyond the degradable/PPT regimes.

\textbf{Open questions and future work.} An important open problem is to prove the spin alignment conjecture with von-Neumann entropy for arbitrary tensor powers.  Although partial progress has been made~\cite{Alhejji_2025}, the general case remains unsolved. Futhermore, understanding the comparative additivity of $D^{(1)}_{\to}$ and $\mathcal{Q}^{(1)}$ remains to be explored. Finally, it remains to be understood whether the techniques developed here extend to two-way distillation or to other quantum resource theories.

The rest of the paper is structured as follows. In Section \ref{sec: prelim}, we provide the necessary preliminary background for our results. In Section \ref{sec:general-principle}, we have identified the general principles that allow us to show single-letter one-way distillable entanglement for non-degradable states. Finally in Section \ref{sec:examples}, we provide explicit examples for such states. Additionally we show that a large class of generalised direct sum channels exhibit single-letter quantum capacity complying with one of our general principles.   


\section{Preliminary}\label{sec: prelim}
\subsection{Bipartite quantum states and entanglement}
Consider two quantum systems with corresponding Hilbert spaces, $\mathcal{H}_A$ and $\mathcal{H}_B$ abbreviated as $A$ (for Alice) and $B$ (for Bob). The states of this composite system, called the bipartite states, are described by elements of $\mathcal{D}(\mathcal{H}_A\otimes \mathcal{H}_B)$, where $\mathcal{D}(\mathcal{H})$ is the set of all positive semi-definite operators on the Hilbert space $\mathcal{H}$ with unit trace. 

The set of bipartite states, $\mathcal{D}(\mathcal{H}_A\otimes \mathcal{H}_B)$ is further classified into two categories: \textit{separable states} and \textit{entangled states}. We call $\varrho_{AB}\in \mathcal{D}(\mathcal{H}_A\otimes \mathcal{H}_B)$ a separable state if there exists a probability vector $\{p_i\}$ and states $\rho_i\in \mathcal{D}(\mathcal{H}_A)$, $\sigma_i\in \mathcal{D}(\mathcal{H}_B)$ such that
\begin{align}
    \varrho_{AB} = \sum_{i} p_i\;\rho_i\otimes\sigma_i
\end{align}
The set of all separable states are often denoted as $\rm{Sep}(A:B)$. The states which cannot be written in this form are called \textit{entangled states}. These are elements of $\mathcal{D}(\mathcal{H}_A\otimes \mathcal{H}_B)\;\backslash\; \rm{Sep}(A:B)$. 

\subsection{Distillable entanglement}
Let \(\rho_{AB}\) be a bipartite state on a finite-dimensional Hilbert space
\(\mathcal{H}_A \otimes \mathcal{H}_B\).
We denote by \(\mathrm{LOCC}_{A\to B}\) the set of protocols implementable by
local operations and classical communication from \(A\) to \(B\) only.
For an integer \(M\), let
\(\Phi_M := \frac{1}{\sqrt{M}}\sum_{i=1}^{M} \ket{i}\ket{i}\)
be a maximally entangled state of Schmidt rank \(M\) on
\(\mathbb{C}^M \otimes \mathbb{C}^M\).
From here onward all the logarithms are taken to be base \(2\).

\paragraph{Achievable rate.}
A number \(R \ge 0\) (in ebits per copy) is said to be \emph{one-way
achievable} for \(\rho_{AB}\) if there exist a sequence of integers \(M_n\)
and a sequence of one-way LOCC protocols
\(\Lambda_n \in \mathrm{LOCC}_{A\to B}\) such that
\[
  \sigma_n := \Lambda_n \big(\rho_{AB}^{\otimes n}\big) \in
  \mathcal{D}\big((\mathbb{C}^{M_n}\!\otimes\!\mathbb{C}^{M_n})\big)
\]
satisfies
\[
  \lim_{n\to\infty} \, \big\| \sigma_n - \Phi_{M_n} \big\|_1 = 0
  \qquad\text{and}\qquad
  \liminf_{n\to\infty} \frac{1}{n}\log M_n \;\ge\; R .
\]

\paragraph{One-way distillable entanglement.}
The \emph{one-way distillable entanglement} of \(\rho_{AB}\) is
\begin{equation}\label{equivalent def:rate}
      D_{\to}(\rho_{AB})
  \;:=\;
  \sup\bigl\{ R \,\big|\, R \text{ is one-way achievable for } \rho_{AB} \bigr\}.
\end{equation}
\paragraph{Equivalent fidelity formulation.}
Equivalently, for \(\varepsilon\in(0,1)\) and \(n\in\mathbb{N}\) define the
\(\varepsilon\)-error, \(n\)-blocklength one-way distillable entanglement by
\begin{align*}
      D_{\to}^{(n,\varepsilon)}(\rho_{AB})
  := \frac{1}{n}\,
  \max_{\Lambda_n \in \mathrm{LOCC}_{A\to B}}
  \{ \log M : F\bigl(\Lambda_n(\rho_{AB}^{\otimes n}),\, \Phi_M\bigr) \ge 1-\varepsilon \Bigr\},
\end{align*}
where \(F(\cdot,\cdot)\) is the Uhlmann fidelity. Then
\begin{equation}\label{equivalent def:n block}
      D_{\to}(\rho_{AB})
  \;=\;
  \lim_{\varepsilon\to 0}\; \liminf_{n\to\infty}
  D_{\to}^{(n,\varepsilon)}(\rho_{AB}) .
\end{equation}

\paragraph{Regularization expression} Devetak and Winter \cite{Devetak_2005} proved the following regularized expression for the one-way distillable entanglement 
\begin{equation}\label{equivalent def:regularized}
    D_{\to}(\rho_{AB})
= \lim_{n\to\infty} \frac{1}{n}\, D^{(1)}_{\to}\!\left(\rho_{AB}^{\otimes n}\right),
\end{equation}
where 
\begin{align}\label{def:one shot}
    D^{(1)}_{\to}(\rho_{AB})
:= \max_{\mc T}\; \sum_m \lambda_m\, I(A'\rangle B)_{\rho_m} 
= \max_{\mc T}\, I(A'\rangle BM)_{\mc T(\rho_{AB})}, 
\quad
\rho_m := \frac{1}{\lambda_m} \, \mc T_m(\rho_{AB}),
\quad
\lambda_m := \operatorname{Tr}\!\left[\mc T_m(\rho_{AB})\right],
\end{align}
where the instrument \(\mc T = \sum_m \mc T_m \otimes \ket{m}\!\bra{m}\) has a single Kraus operator
\(\mc T_m(\cdot) = K_m (\cdot) K_m^\dagger\), \(K_m: A \to A'\).

We can obtain an alternative expression in terms of an isometric extension \(V: A \to A'MN\) of the instrument, $\mathcal{T}$, defined by
\begin{equation}\label{eqn:instrument isometry}
    V := \sum_m K_m \otimes \ket{m}_M \otimes \ket{m}_N,
\end{equation}
for a classical register \(N \cong M\).
Since \(\sum_m K_m^\dagger K_m = \mathbb{I}_A\), we have \(V^\dagger V = \mathbb{I}_A\),
and \(\mc T(\rho_{AB}) = \operatorname{Tr}_N\!\big(V \rho_{AB} V^\dagger\big)\). For any isometry $V^{A\to A'MN}$ with the form~\eqref{eqn:instrument isometry}, and let $\rho_{AB}$ be a bipartite state with purification $\ket{\phi}_{ABE}$, define 
\begin{equation}
    \ket{\omega}_{A'MN BE} = (V \otimes \mb I_B \otimes \mb I_E) \, \ket{\phi}_{ABE},
\end{equation}
we have 
\begin{align*}
    I(A'\rangle BM)_{\mc T(\rho_{AB})}& = I(A'\rangle BM)_\omega \\
    & = S(BM)_\omega - S(A'BM)_\omega \\
    & = S(BM)_\omega - S(EN)_\omega \\
    & = S(BM)_\omega - S(EM)_\omega,
\end{align*}
where the third equality follows from the fact that $\omega$ is pure on $A'MNBE$, so $S(A'B M)_\omega = S(EN)_\omega$ (complementary systems with a joint pure state have equal von Neumann entropy); the last equality follows from $M\!\leftrightarrow\!N$ symmetry in the construction of $\omega$ (in particular, the joint law of $(E,M)$ and $(E,N)$ is the same), hence $S(EN)_\omega = S(EM)_\omega$. Therefore, we have 
\begin{lemma}\label{lemma:alternative}
For any bipartite state $\rho_{AB}$ with purification $\ket{\phi}_{ABE}$, we have
\begin{equation}
    D^{(1)}_{\to}(\rho_{AB})
= \max_V\, I(A'\rangle BM)_\omega = \max_V\, \left[S(BM)_\omega - S(EM)_\omega\right],\quad \ket{\omega}_{A'MN BE} = (V \otimes \mb I_B \otimes \mb I_E) \, \ket{\phi}_{ABE},
\end{equation}
where $V^{A\to A'MN}$ is any isometry with the form~\eqref{eqn:instrument isometry}.
\end{lemma}
From the above lemma, it is immediate to see that if $U_A$ and $V_B$ are isometries, then
\begin{equation}\label{eqn:unitary invariance}
    D^{(1)}_{\to}\bigl((U_A\otimes V_B)\,\rho\,(U_A^\dagger\otimes V_B^\dagger)\bigr)=D^{(1)}_{\to}(\rho),\quad D_{\to}\bigl((U_A\otimes V_B)\,\rho\,(U_A^\dagger\otimes V_B^\dagger)\bigr)=D_{\to}(\rho)
\end{equation}
for every $\rho\in\mc D(AB)$. 

\subsection{Degradable states}
\begin{definition}
Let $\rho_{AB}$ be a bipartite state with purification $\ket{\phi}_{ABE}$. 
The state $\rho_{AB}$ is called:
\begin{enumerate}
\item \emph{degradable}, if there is a quantum channel $\mc D^{B\to E}$ such that 
\begin{equation}\label{eq:degradable-equalities}
(id\otimes \mc D^{B \to E})(\rho_{AB})=\phi_{AE}.
\end{equation}


\item \emph{antidegradable}, if there is a quantum channel $\mc D^{E\to B}$ such that 
\begin{equation}
    (id\otimes \mc D^{E\to B})(\phi_{AE})=\rho_{AB}.
\end{equation}
\end{enumerate}
\end{definition}
Via data processing inequality, for the state given by $\ket{\omega}_{A'MN BE} = (V \otimes \mb I_B \otimes \mb I_E)\ket{\phi}_{ABE}$, we have $S(BM)_\omega - S(EM)_\omega \le S(B)_\omega - S(E)_\omega = I(A\rangle B)_\rho$. Then using Lemma \ref{lemma:alternative}, we find that $D_\to^{(1)}(\rho) =I(A\rangle B)_\rho$ for degradable state $\rho$. Similarly, $D_\to^{(1)}(\rho) = 0$ for anti-degradable state $\rho$. This provides an alternative proof of the results obtained in \cite{Leditzky_2018}. Using the additivity of coherent information under tensor product states, we immediately have 
\begin{align}
    D_\to (\rho_{AB}) = \begin{cases}
        I(A\rangle B)_\rho, &\text{ when $\rho$ is degradable};\\
        0, &\text{ when $\rho$ is anti-degradable}.
    \end{cases}
\end{align}

\subsection{Connection to the quantum channel capacities}
To study the additivity property of $D^{(1)}_\to$, we introduce an upper bound of it. Recall \eqref{def:one shot}:
\begin{align*}\label{def:D1_opt}
    D^{(1)}_\to(\rho_{AB}) = \max \left\{I(A'\rangle BM)_{\omega_{A'BM}}: \omega_{A'BM} = \sum_m (K_m \otimes \mb I_B) \rho_{AB} (K_m \otimes \mb I_B)^\dagger \otimes |m\rangle \langle m|,\quad \forall \{K_m^{A \to A'}\}\ \text{s.t.}\ \sum_m K_m^\dagger K_m = \mb I_A\right\}.
\end{align*}
Therefore, a natural upper bound for $D^{(1)}_\to$ is to withdraw the condition $\sum_m K_m^\dagger K_m = \mb I_A$: 

\begin{definition}\label{def: Non-TP instrument}
    For a bipartite state $\rho_{AB}$, we introduce the quantity 
    \begin{align*}
    \widehat D^{(1)}_\to(\rho_{AB}) = \max \left\{I(A'\rangle BM)_{\omega_{A'BM}}: \omega_{A'BM} = \frac{\sum_m (K_m \otimes \mb I_B) \rho_{AB} (K_m \otimes \mb I_B)^\dagger \otimes |m\rangle \langle m|}{\sum_m \Tr( (K_m^\dagger K_m \otimes \mb I_B)\rho_{AB})},\quad \forall \{K_m\}\right\}.
\end{align*}
\end{definition}
Interestingly, the above quantity is actually the maximal coherent information of the channel with Choi operator given by $\rho_{AB}$:
\begin{prop} \label{prop:connection to capacity}
Let $\rho_{AB}$ be bipartite state and $\mc{N}$ be the corresponding completely positive map from $A$ to $B$. Then, 
    \begin{align*}
    \mc Q^{(1)}(\mc N) =  \widehat D^{(1)}_\to(\rho_{AB}).
\end{align*}
\end{prop}

A first step is to simplify the above expression. 
\begin{lemma}
\begin{equation}\label{def:single operator}
            \begin{aligned}
    \widehat D^{(1)}_\to(\rho_{AB}) = \max \left\{I(A'\rangle B)_{\omega_{A'B}}: \omega_{A'B} = \frac{(K \otimes \mb I_B) \rho_{AB} (K \otimes \mb I_B)^\dagger }{\Tr( (K^\dagger K \otimes \mb I_B)\rho_{AB})}, \quad \forall K: A\to A'\right\}.
\end{aligned}
\end{equation}
\end{lemma}
\begin{proof}
    It is obvious that         \begin{align*}
    \widehat D^{(1)}_\to(\rho_{AB}) \ge \max \left\{I(A'\rangle B)_{\omega_{A'B}}: \omega_{A'B} = \frac{(K \otimes \mb I_B) \rho_{AB} (K \otimes \mb I_B)^\dagger}{\Tr( (K^\dagger K \otimes \mb I_B)\rho_{AB})}, \quad \forall K: A\to A'\right\}.
\end{align*}
To show the other direction, suppose $\{K_m\}$ are the optimal operators achieving $\widehat D^{(1)}_\to(\rho_{AB})$. Denote 
\begin{align*}
    p_m = \frac{\Tr( (K_m^\dagger K_m \otimes \mb I_B)\rho_{AB})}{\sum_m \Tr( (K_m^\dagger K_m \otimes \mb I_B)\rho_{AB})},\quad \rho_{A'B}^m = \frac{(K_m \otimes \mb I_B) \rho_{AB} (K_m \otimes \mb I_B)^\dagger}{\Tr( (K_m^\dagger K_m \otimes \mb I_B)\rho_{AB})}
\end{align*}
Then the coherent information of $\omega_{A'BM} = \sum_m p_m \rho_{A'B}^m \otimes |m\rangle \langle m|$ is calculated as 
\begin{align*}
    I(A'\rangle BM)_{\omega_{A'BM}} = \sum_m p_m I(A'\rangle B)_{\rho_{A'B}^m} \le I(A'\rangle B)_{\rho_{A'B}^{m_0}},
\end{align*}
where $m_0$ is chosen such that 
\begin{align*}
    I(A'\rangle B)_{\rho_{A'B}^{m_0}} = \max_m I(A'\rangle B)_{\rho_{A'B}^m}.
\end{align*}
Therefore, by choosing $K$ as $K_{m_0}$, one has
\begin{align*}
    \widehat D^{(1)}_\to(\rho_{AB}) \le \max \left\{I(A'\rangle B)_{\omega_{A'B}}: \omega_{A'B} = \frac{1}{\Tr( (K^\dagger K \otimes \mb I_B)\rho_{AB})} (K \otimes \mb I_B) \rho_{AB} (K \otimes \mb I_B)^\dagger \quad \forall K: A\to A'\right\},
\end{align*}
which concludes the proof.
\end{proof}
For any bipartite state $\rho_{AB}$, it naturally induces a completely positive map $\mc N^{A\to B}$ via the following
\begin{equation}
    \mc N^{A\to B}(\sigma):= d_A\, \Tr_A((\sigma^T\otimes \mb I_B)\rho_{AB}),\quad \rho_{AB} = d_A \, (id \otimes \mc N^{A\to B})(|\Phi\rangle \langle \Phi|).
\end{equation}
For any completely positive map $\mc N$, one can define its maximal coherent information by
\begin{equation}\label{def:channel coherent information}
    \mc Q^{(1)}(\mc N):= \max \left\{I(A'\rangle B)_{\omega_{A'B}}: \omega_{A'B} = \frac{(id \otimes \mc N)(|\psi \rangle \langle \psi |_{A'A})}{\Tr((id \otimes \mc N)(|\psi \rangle \langle \psi |_{A'A}))},\quad \forall |\psi\rangle_{A'A} \right\}.
\end{equation}

\begin{proof}[Proof of Proposition~\ref{prop:connection to capacity}]
    Recall that there is a one-to-one correspondence (up to normalization) between pure states $|\psi\rangle_{A'A}$ on $A'A$ and operators $K:A \to A'$ such that 
\begin{align*}
    |\psi\rangle_{A'A} = (K \otimes \mb I_A) |\Phi\rangle_{AA}.
\end{align*}
Note that \begin{align*}
    (K \otimes \mb I_B) \rho_{AB} (K \otimes \mb I_B)^\dagger = (K \otimes \mb I_B) (id \otimes \mc N)(|\Phi\rangle \langle \Phi |) (K \otimes \mb I_B)^\dagger = (id \otimes \mc N)(|\psi \rangle \langle \psi |_{A'A})),
\end{align*}
thus using the definitions~\eqref{def:channel coherent information} and \eqref{def:single operator}, we conclude the proof.
\end{proof}

\begin{lemma}
    Let $\rho_{AB}$ be a bipartite state. Then we have, $\widehat D^{(1)}_\to(\rho_{AB})=0 \iff D^{(1)}_\to(\rho_{AB})=0$ 
\end{lemma}
\begin{proof}
    As $\widehat D^{(1)}_\to(\rho_{AB})\geq D^{(1)}_\to(\rho_{AB})\geq 0$, it trivially follows that $\widehat D^{(1)}_\to(\rho_{AB}) = 0 \implies D^{(1)}_\to(\rho_{AB})=0$. To show the other way around, assume that $\widehat D^{(1)}_\to(\rho_{AB})> 0$. Therefore, there exists $K:A\to A'$ such that 
    $$I(A'\rangle B)_{\omega_{A'B}}>0;\qquad \omega_{A'B} = \frac{1}{\Tr( (K^\dagger K \otimes \mb I_B)\rho_{AB})} (K \otimes \mb I_B) \rho_{AB} (K \otimes \mb I_B)^\dagger$$
    Now choose $0<c\leq 1$ such that $E_0 :=c(K^\dagger K)\leq \mb{I}_A$. We further perform the following spectral decomposition 
    $$\mb{I}_A - E_0 = \sum_{k=1}^{\mathrm{dim} A} E_k$$
    Here $E_k$ for $k\geq 1$ are rank-1 positive operators with $0\leq \mathrm{Tr} (E_k) \leq 1$. Firstly, we note that $\{E_k\}_{k=0}^{\mathrm{dim} A}$ forms a POVM, therefore defines an instrument, $\mathcal{T}:A\to A'M$ where $M\cong\mb{C}^{\rm{dim} A +1}$. Secondly, for $k\geq 1$, define Kraus operators, $K_k:=U\sqrt{E_k} = \ketbra{\beta}{\alpha}$ for some vectors $\ket{\alpha}\in A, \ket{\beta}\in A'$, then the state
    $$\omega^{(k)}_{A'B}:=\ketbra{\beta}{\beta}_{A'}\otimes \frac{1}{\bra{\alpha} \rho_A\ket{\alpha}} \mathrm{Tr}_A\left[(\ketbra{\alpha}{\alpha}\otimes \mb{I}_B)\rho_{AB}\right]$$
    is a product state on $A'B$, i.e., having $I(A'\rangle B)_{\omega^{(k)}_{A'B}} = 0$. Therefore, we have
    \begin{align}
        I(A'\rangle BM)_{\mc{T}(\rho)} = c \Tr( (K^\dagger K \otimes \mb I_B)\rho_{AB}) \;I(A'\rangle B)_{\omega_{A'B}} > 0 \implies D^{(1)}_{\to}(\rho_{AB}) >0
    \end{align}
    By contrapositivity this is equivalent to $D^{(1)}_\to(\rho_{AB})=0 \implies \widehat D^{(1)}_\to(\rho_{AB}) = 0 $
\end{proof}


\section{General principle of non-degradable states with
single-letter one-way distillable entanglement}\label{sec:general-principle}
It is well-known that the one-way distillable entanglement of (anti-)degradable states has a single-letter expression (see \cite{Leditzky_2018}) and is equal to the coherent information. Moreover, two-way (thus one-way) distillable entanglement of PPT states is known to be zero. In this section, we introduce three distinct classes of non-degradable states that exhibit additivity in the one-way distillable entanglement. The first class relies on a weaker notion of degradable states. The second class concerns with states having an useless component. Thirdly, we propose a special class of generalised direct sum states which has a single-letter one-way distillable entanglement.


\color{black}
\subsection{States with weaker notion of degradability}
Motivated from previous work \cite{Smith_2025}, we propose the following two notions of weaker degradability
\begin{definition}
    We say that the bipartite state $\rho_{AB}$ with purification $\ket{\phi}_{ABE}$ is 
    \begin{itemize}
        \item less noisy at level $n$, if for any quantum instrument $\mc T^{A^n \to A'M}$, for the state $\omega_{A'B^nE^n}= \mc T^{A^n \to A'M}(\phi_{ABE}^{\otimes n})$, we have 
        \begin{equation}\label{less noisy level n}
            I(M;B^n) \ge I(M;E^n).
        \end{equation}
        We say $\rho_{AB}$ is regularized less noisy, if \eqref{less noisy level n} holds for any $n \ge 1$.
        \item informationally degradable, if for any quantum channel $\mc N^{A \to A'}$, for the state $\omega_{A'BE}= \mc N^{A \to A'}(\phi_{ABE})$, we have 
        \begin{equation}\label{informational degradable}
            I(A';B) \ge I(A';E). 
        \end{equation}
    \end{itemize}
\end{definition}
\begin{prop}\label{prop:single-letter}
Let $\rho_{AB}$ be a bipartite state which is less noisy at level $n \ge 1$. Then we have 
\begin{equation}
    D_\to^{(1)}(\rho_{AB}^{\otimes n}) = n D_\to^{(1)}(\rho_{AB}) = n\, I(A \rangle B).
\end{equation}
Thus if $\rho_{AB}$ is regularized less noisy, we have $D_\to(\rho_{AB})=D_\to^{(1)}(\rho_{AB})$.
\end{prop}
\begin{proof}
Fix $n\ge 1$ and suppose the optimal quantum instrument for $D^{(1)}_\to(\rho_{AB}^{\otimes n})$ is given by \(\mc T^{A^n \to A'M}\). Denote $V:A^n\to A'MN$ is the isometry of $\mc T$, defined as in \eqref{eqn:instrument isometry}. Via Lemma \ref{lemma:alternative} and the optimality of the quantum instrument,
\begin{align*}
    D^{(1)}_\to (\rho^{\otimes n}_{AB}) = S(B^n M)_\omega - S(E^n M)_\omega,\quad \ket{\omega}_{A'MNB^nE^n} := V^{A^n\to A'MN} \ket{\phi}^{\otimes n}_{ABE}.
\end{align*}
Note that $I(M;B^n)_\omega \ge I(M;E^n)_\omega$ is equivalent to $S(B^n M)_\omega - S(E^n M)_\omega \le S(B^n)_\omega - S(E^n)_\omega$ which implies 
\begin{align*}
    D^{(1)}_\to (\rho^{\otimes n}_{AB}) \le S(B^n)_\omega - S(E^n)_\omega = I(A^n \rangle B^n)_{\rho^{\otimes n}}
     \;=\; n\, I(A\rangle B)_\rho .
\end{align*}
where the last two equalities follow from the additivity of von Neumann entropy over the tensor products: $S(B^n)_{\rho^{\otimes n}} = n\,S(B)_\rho$ and $S(A^nB^n)_{\rho^{\otimes n}} = n\,S(AB)_\rho$, hence $I(A^n\rangle B^n)_{\rho^{\otimes n}} = S(B^n) - S(A^nB^n) = n\,I(A\rangle B)_\rho$. We conclude the proof by recalling that $n\,I(A\rangle B)_\rho$ is a lower bound for $D_\to^{(1)}(\rho_{AB}^{\otimes n})$.

\end{proof}

\begin{prop}
 Let $\rho_{AB}$ and $\sigma_{\wt A\wt B}$ be informationally degradable. Then for all $n,k\ge 0$,
\begin{equation}\label{eq:additivity-ID}
    D_{\to}^{(1)}\!\big(\rho_{AB}^{\otimes n}\!\otimes\!\sigma_{\wt A\wt B}^{\otimes k}\big)
    \;=\; n\,D_{\to}^{(1)}(\rho_{AB}) \;+\; k\,D_{\to}^{(1)}(\sigma_{\wt A\wt B})=  n\, I(A\rangle B)_\rho + k\, I(\wt A \rangle \wt B)_\sigma.
\end{equation}
In particular, informationally degradable states are single-letter for one-way distillable entanglement and are additive under tensor products.
\end{prop}
\begin{proof}
    Let $\ket{\phi}_{ABE}$, $\ket{\wt \phi}_{\wt A \wt B \wt E}$ be the purification states of $\rho,\sigma$ respectively. Fix $n,k \ge 0$ ($0$ means empty here), suppose $$\mc T^{A^n\wt A^k\to A'M}$$
    is the optimal quantum instrument for $D_{\to}^{(1)}\!\big(\rho_{AB}^{\otimes n}\!\otimes\!\sigma_{\wt A\wt B}^{\otimes k}\big)$. Denote $V:A^n\wt A^k\to A'MN$ is the isometry of $\mc T$, defined as in \eqref{eqn:instrument isometry}. Via Lemma \ref{lemma:alternative} and the optimality of the quantum instrument, we have 
    \begin{align*}
        D_{\to}^{(1)}\!\big(\rho_{AB}^{\otimes n}\!\otimes\!\sigma_{\wt A\wt B}^{\otimes k}\big) = S(B^n\wt B^k M)_\omega - S(E^n\wt E^kM)_\omega,
    \end{align*}
    where \begin{align}\label{eqn:state informational degradable}
        \ket{\omega}_{A'MNB^n \wt B^kE^n \wt E^k} := V^{A^n\wt A^k\to A'MN} \otimes \mb I_{B^n \wt B^k} \otimes \mb I_{E^n \wt E^k} \left(\ket{\phi}_{ABE}^{\otimes n} \otimes \ket{\wt \phi}^{\otimes k}_{\wt A \wt B \wt E}\right).
    \end{align}
    Now denote 
    \begin{align*}
        B^n = B_1 B_2 \cdots B_n,\quad \wt B^k = \wt B_1 \wt B_2 \cdots \wt B_k,
    \end{align*}
    and similar for $E^n, \wt E^k$. Then via telescoping argument, 
    \begin{align*}
        S(B^n\wt B^k M)_\omega - S(E^n\wt E^kM)_\omega & = S(B^n\wt B^k M)_\omega - S(B^n\wt E^kM)_\omega + S(B^n\wt E^kM)_\omega - S(E^n\wt E^kM)_\omega \\
        & = \sum_{j=1}^k \left(S(\wt B_j \wt V_jB^n M)_\omega  -  S(\wt E_j \wt V_jB^n M)_\omega \right) + \sum_{i=1}^n \left( S(B_i V_i \wt E^kM)_\omega  - S(E_i V_i \wt E^kM)_\omega \right),
    \end{align*}
    where $V_i, \wt V_j$ are $n-1$, $k-1$ subsystems defined by 
    \begin{align*}
        V_i =
\begin{cases}
B_2 \cdots B_n, & i=1,\\[2pt]
E_1 \cdots E_{i-1}\, B_{i+1} \cdots B_n, & 2 \le i \le n-1,\\[2pt]
E_1 \cdots E_{n-1}, & i=n.
\end{cases},\quad \wt V_j = \begin{cases}
\wt B_2 \cdots \wt B_k, & j=1,\\[2pt]
\wt E_1 \cdots \wt E_{j-1}\, \wt B_{j+1} \cdots \wt B_k, & 2 \le j \le k-1,\\[2pt]
\wt E_1 \cdots \wt E_{k-1}, & j=k.
\end{cases}
    \end{align*}
    Now we claim that 
\begin{align}\label{claim:informational degradable}
    S(\wt B_j \wt V_jB^n M)_\omega  -  S(\wt E_j \wt V_jB^n M)_\omega \le S(\wt B_j)_\omega - S(\wt E_j)_\omega,\quad S(B_i V_i \wt E^kM)_\omega  - S(E_i V_i \wt E^kM)_\omega \le S(B_i)_\omega - S(E_i)_\omega,\ \forall i,j. 
\end{align}
In fact, for each $1\le j \le k$, one can construct a quantum channel $\wt{\mc N_j}^{\wt A_j \to \wt V_j B^n M}$ by 
\begin{align*}
    \wt{\mc N_j}(\rho_{\wt A_j}):= \Tr_{(\wt V_j B^n M)^c} \bigg[ V^{A^n\wt A^k\to A'MN} (\ketbra{\phi}{\phi}_{ABE}^{\otimes n} \otimes \rho_{\wt A_j} \otimes  \ketbra{\wt \phi}{\wt \phi}^{\otimes (k-1)}_{\wt A \wt B \wt E}) (V^{A^n\wt A^k\to A'MN})^{\dagger}\bigg],
\end{align*}
where $(\wt A_j \wt V_j B^n M)^c$ denotes the complement of $\wt A_j \wt V_j B^n M$. Therefore, we have 
\begin{align*}
    & \wt{\mc N_j}^{\wt A_j \to \wt V_j B^n M} (\ketbra{\wt \phi}{\wt \phi}_{\wt A_j \wt B_j \wt E_j}) = \omega_{\wt B_j \wt E_j \wt V_j B^n M},  
\end{align*}
where $\omega_{\wt B_j \wt E_j \wt V_j B^n M}$ is the reduced density of $\omega$ defined in \eqref{eqn:state informational degradable}. Since $\ket{\wt \phi}_{\wt A_j \wt B_j \wt E_j}$ is the purification of the informational degradable state $\sigma_{\wt A\wt B}$, via \eqref{informational degradable}, we have 
\begin{align*}
    I(\wt V_jB^n M; \wt B_j )\ge I(\wt V_jB^n M;\wt E_j),
\end{align*}
concluding the claim that $S(\wt B_j \wt V_jB^n M)_\omega  -  S(\wt E_j \wt V_jB^n M)_\omega \le S(\wt B_j)_\omega - S(\wt E_j)_\omega$. Similarly, we have $$S(B_i V_i \wt E^kM)_\omega  - S(E_i V_i \wt E^kM)_\omega \le S(B_i)_\omega - S(E_i)_\omega. $$ Therefore, one has 
\begin{align*}
    S(B^n\wt B^k M)_\omega - S(E^n\wt E^kM)_\omega & = \sum_{j=1}^k \left(S(\wt B_j \wt V_jB^n M)_\omega  -  S(\wt E_j \wt V_jB^n M)_\omega \right) + \sum_{i=1}^n \left( S(B_i V_i \wt E^kM)_\omega  - S(E_i V_i \wt E^kM)_\omega \right) \\
    & \le k (S(\wt B)_\sigma - S(\wt E)_\sigma) + n(S(B)_\rho - S(E)_\rho) = k\, I(\wt A \rangle \wt B)_\sigma + n\, I(A\rangle B)_\rho.
\end{align*}
We conclude the proof by recalling that $D_{\to}^{(1)}\!\big(\rho_{AB}^{\otimes n}\!\otimes\!\sigma_{\wt A\wt B}^{\otimes k}\big) \ge k\, I(\wt A \rangle \wt B)_\sigma + n\, I(A\rangle B)_\rho$.
\end{proof}


\subsection{Non-degradable states with useless component}
In this subsection we exhibit a family of generally non-degradable states whose one-way distillable entanglement is single-letter. Heuristically, these states are orthogonally flagged mixtures on Alice, so the additivity inherited from the components persists even though degradability may fail. The formal statement is:
\begin{theorem}\label{main:class 1}
    Let $\rho_0$ and $\rho_1$ be bipartite states on $AB$ satisfying \begin{equation}\label{assumption:additivity}
    D^{(1)}_{\to}(\bigotimes_{i=1}^n \rho_{w_i})
\;\le\; \sum_{i=1}^n D^{(1)}_{\to}(\rho_{w_i})
\;=\; (n-|w^n|)\,D^{(1)}_{\to}(\rho_0)+|w^n|\,D^{(1)}_{\to}(\rho_1)
\end{equation}
for all $w^n = (w_1,\cdots w_n)\in\{0,1\}^n$ and $n\in\mathbb{N}$. Moreover, assume $\rho_0^A \perp \rho_1^A$. Then for all $p\in[0,1]$, we have 
    \begin{align*}
        D_{\to}\big(p\rho_0+(1-p)\rho_1\big) = D^{(1)}_{\to}\big(p\rho_0+(1-p)\rho_1\big)= p D_{\to}(\rho_0) + (1-p) D_{\to}(\rho_1).
    \end{align*}
\end{theorem}
To prove the above theorem, we need the following lemmas:
\begin{lemma}\label{lemma:concavity}
Let $\sigma,\tau$ be bipartite states on $AB$ such that $\sigma^A \perp \tau^A$, i.e., $\mathrm{supp}(\sigma^A) \perp \mathrm{supp}(\tau^A)$. Then for $p\in[0,1]$, we have 
\begin{equation}
   D^{(1)}_{\to}(\rho) \ge \;p\,D^{(1)}_{\to}(\sigma) + (1-p) D^{(1)}_{\to}(\tau),\quad  D_{\to}(\rho) \ge \;p\,D_{\to}(\sigma) + (1-p) D_{\to}(\tau).
\end{equation}
\end{lemma}
\begin{proof}
Define $A_0 := \mathrm{supp}(\sigma^A), A_1:= \mathrm{supp}(\tau^A)$, we have $\mathrm{supp}(\sigma)\subseteq A_0B,\ \mathrm{supp}(\tau)\subseteq A_1B$. Define an isometry $V:A\to X\otimes A'$:
\[
V\big|_{A_0}=|0\rangle_X\otimes id_{A_0},\qquad
V\big|_{A_1}=|1\rangle_X\otimes id_{A_1},
\]
where $X=\mathrm{span}\{|0\rangle,|1\rangle\}$ and $A'$ is a copy of $A_0\oplus A_1$.
Let
\[
\omega\;:=\;(V\otimes \mb I_B)\,\rho\,(V^\dagger\otimes \mb I_B)
\;=\; p\,|0\rangle\!\langle0|_X\otimes \sigma\;+\;(1-p)\,|1\rangle\!\langle1|_X\otimes \tau.
\]
Via \eqref{eqn:unitary invariance}, we have $D^{(1)}_{\to}(\omega)=D^{(1)}_{\to}(\rho),\ D_{\to}(\omega)=D_{\to}(\rho)$. It remains to show that 
\begin{align}
& D^{(1)}_{\to}(\omega) = D^{(1)}_{\to}(p\,|0\rangle\!\langle0|_X\otimes \sigma\;+\;(1-p)\,|1\rangle\!\langle1|_X\otimes \tau) \ge p\,D^{(1)}_{\to}(\sigma) + (1-p) D^{(1)}_{\to}(\tau), \label{concavity:1}\\
    & D_{\to}(\omega) = D_{\to}(p\,|0\rangle\!\langle0|_X\otimes \sigma\;+\;(1-p)\,|1\rangle\!\langle1|_X\otimes \tau) \ge p\,D_{\to}(\sigma) + (1-p) D_{\to}(\tau).\label{concavity:regularized}
\end{align}
To show~\eqref{concavity:1}, suppose $\mc T^0:A \to M_0A_0',\ \mc T^1: A \to M_1 A_1'$ are the optimal quantum instruments for $D^{(1)}_{\to}(\sigma)$ and $D^{(1)}_{\to}(\tau)$, respectively. Then one finds the quantum instrument $\mc T = \ketbra{0}{0} \cdot \ketbra{0}{0} \otimes \mc T^0 + \ketbra{1}{1} \cdot \ketbra{1}{1} \otimes \mc T^1$ provides a lower bound of $D^{(1)}_{\to}(p\,|0\rangle\!\langle0|_X\otimes \sigma\;+\;(1-p)\,|1\rangle\!\langle1|_X\otimes \tau)$, which is $p\,D^{(1)}_{\to}(\sigma) + (1-p) D^{(1)}_{\to}(\tau)$. 

To show~\eqref{concavity:regularized}, for each $n$, $\omega^{\otimes n}$ is block-diagonal in the computational basis of $X^n$:
\[
\omega^{\otimes n}\;=\;\sum_{x^n\in\{0,1\}^n} p^{N_0(x^n)}(1-p)^{N_1(x^n)}\,
|x^n\rangle\!\langle x^n|_{X^n}\ \otimes\ \sigma^{\otimes N_0(x^n)} \otimes \tau^{\otimes N_1(x^n)},
\]
where $N_0(x^n)$ and $N_1(x^n)=n-N_0(x^n)$ denote the type counts (numbers of 0's and 1's).
Let $\Delta$ be the completely dephasing channel on $X^n$ in this basis; then $\Delta(\omega^{\otimes n})=\omega^{\otimes n}$.
Therefore prepending $\Delta$ to any protocol does not change its action on $\omega^{\otimes n}$.
Hence, without loss of generality, any protocol may be assumed to \emph{measure} $X^n$ in this basis at the beginning and obtain the classical string $x^n$.

Fix $\varepsilon>0$.
By the definition of $D_{\to}(\sigma), \ D_{\to}(\tau)$, there exists $K\ge 1$ such that for all $N\ge K$ there are one-way protocols $\Pi^{\sigma}_N, \Pi^{\tau}_N$ on $\sigma^{\otimes N}$ and $\tau^{\otimes N}$ producing $(N,M^{\sigma}_N,\varepsilon)$ and $(N,M^{\tau}_N,\varepsilon)$ with
\begin{equation}\label{eqn:protocol}
    \frac{1}{N}\log M^\sigma_N \;\ge\; D_{\to}(\sigma)-\varepsilon,\quad \frac{1}{N}\log M^\tau_N \;\ge\; D_{\to}(\tau)-\varepsilon.
\end{equation}
On input $\omega^{\otimes n}$, measure $X^n$ to obtain $x^n$ and the conditional state
$\sigma^{\otimes N_0(x^n)}\otimes \tau^{\otimes N_1(x^n)}$.
If $N_0(x^n),N_1(x^n)\ge K$, run $\Pi^{\sigma}_{N_0(x^n)}$ on the $\sigma$-positions and $\Pi^{\tau}_{N_1(x^n)}$ on the $\tau$-positions; otherwise output a trivial product (zero ebits). For large enough $n$,
\begin{align*}
    n\,D^{(n,\varepsilon)}_\to(\omega) & \ge \sum_{x^n\in\{0,1\}^n, N_0(x^n),N_1(x^n) \ge K} p^{N_0(x^n)}(1-p)^{N_1(x^n)} \left(\log M^\sigma_{N_0(x^n)} + \log M^\tau_{N_1(x^n)}\right) \\
    & \ge \sum_{x^n\in\{0,1\}^n, N_0(x^n),N_1(x^n) \ge K} p^{N_0(x^n)}(1-p)^{N_1(x^n)} \left( N_0(x^n) (D_\to(\sigma) - \varepsilon) + N_1(x^n) (D_\to(\tau) - \varepsilon)\right),
\end{align*}
where for the first inequality, we use the one-way protocols consisting of $\Pi^{\sigma}_{N_0(x^n)}, \Pi^{\tau}_{N_1(x^n)}$ depending on the measurement outcome on $X^n$, and the second inequality follows from~\eqref{eqn:protocol}. Finally, using 
\begin{align*}
    & \sum_{x^n\in\{0,1\}^n, N_0(x^n),N_1(x^n) \ge K} p^{N_0(x^n)}(1-p)^{N_1(x^n)} N_0(x^n) = pn(1- o_n(1)), \\
    & \sum_{x^n\in\{0,1\}^n, N_0(x^n),N_1(x^n) \ge K} p^{N_0(x^n)}(1-p)^{N_1(x^n)} N_1(x^n) = (1-p)n(1- o_n(1)),
\end{align*}
we have \begin{align*}
    D^{(n,\varepsilon)}_\to(\omega) \ge p(1-o_n(1)) (D_\to(\sigma) -\varepsilon) +(1-p)(1-o_n(1)) (D_\to(\tau) -\varepsilon)
\end{align*}
Let $n\to \infty$ and then $\varepsilon \to 0$, by definition~\eqref{equivalent def:n block} we have 
\begin{align*}
    D_{\to}(\rho)=D_{\to}(\omega)  \ge \;p\,D_{\to}(\sigma) + (1-p) D_{\to}(\tau).
\end{align*}
\end{proof}
For the other direction, we use the following lemma, which is Proposition 2.7 in \cite{Leditzky_2018}:
\begin{lemma}\label{lemma:convexity}
Let $\rho_0$ and $\rho_1$ be bipartite states on $AB$ satisfying
\begin{equation}
    D^{(1)}_{\to}(\bigotimes_{i=1}^n \rho_{w_i})
\;\le\; \sum_{i=1}^n D^{(1)}_{\to}(\rho_{w_i})
\;=\; (n-|w^n|)\,D^{(1)}_{\to}(\rho_0)+|w^n|\,D^{(1)}_{\to}(\rho_1)
\end{equation}
for all $w^n = (w_1,\cdots w_n)\in\{0,1\}^n$ and $n\in\mathbb{N}$. Then for all $p\in[0,1]$,
\[
D_{\to}\big(p\rho_0+(1-p)\rho_1\big)
\;\le\; p\,D_{\to}(\rho_0)+(1-p)\,D_{\to}(\rho_1).
\]
\end{lemma}
We are able to prove Theorem~\ref{main:class 1}:
\begin{proof}[Proof of Theorem~\ref{main:class 1}]
    The first conclusion follows directly from Lemma~\ref{lemma:concavity} and Lemma~\ref{lemma:convexity}. For the second conclusion (additivity), we have 
    \begin{align*}
        p\,D^{(1)}_{\to}(\rho_0) + (1-p) D^{(1)}_{\to}(\rho_1) &\underset{\text{Lemma~\ref{lemma:concavity}}}{\le}D^{(1)}_{\to}\big(p\rho_0+(1-p)\rho_1\big) \le  D_{\to}\big(p\rho_0+(1-p)\rho_1\big) \\
        & \underset{\text{Lemma~\ref{lemma:convexity}}}{\le} p\,D_{\to}(\rho_0)+(1-p)\,D_{\to}(\rho_1) =  p\,D^{(1)}_{\to}(\rho_0)+(1-p)\,D^{(1)}_{\to}(\rho_1).
    \end{align*}
\end{proof}
\subsection{Spin alignment phenomenon}
In this subsection we introduce the spin-alignment phenomenon, generalized from~\cite{Leditzky_2023}.
Let $\sigma_0\in\mc D(B_0)$ and $\sigma_1\in\mc D(B_1)$ be fixed states, and assume $A_i\simeq B_i$ for $i=0,1$.
Define the (partially erasing) channels
\begin{align}
  \mc N_0^{A_0\to B_0B_1}(\rho) &:= \rho\otimes \sigma_1,\\
  \mc N_1^{A_1\to B_0B_1}(\rho) &:= \sigma_0\otimes \rho.
\end{align}
Given $n\ge 1$ and suppose $\vec x=(x_1,\dots,x_n)\in\{0,1\}^n$ is a $n$-bit string. Define the product channel
\begin{equation}
      \mc N_{\vec x}:= \bigotimes_{i=1}^n \mc N_{x_i}^{A_{x_i}^{(i)}\to B_0^{(i)}B_1^{(i)}}.
\end{equation}
For each position $i$, the input system is $A_{x_i}^{(i)}$ and the output system is $B_0^{(i)}B_1^{(i)}$.
For $n$-bit string $\vec x$, the input state is $\rho_{\vec x}\in\mc D \Big(\bigotimes_{i=1}^n A_{x_i}^{(i)}\Big)$. The spin alignment phenomenon is as follows:
\begin{conjecture}[Spin alignment phenemenon]\label{conj:spin-alignment}
Fix a probability distribution $\{p_{\vec x}\}_{\vec x\in\{0,1\}^n}$, the minimization of the von Neumann entropy 
\begin{equation}\label{eqn:spin von Neumann}
  \min_{\rho_{\vec x}} S \left(\sum_{\vec x\in\{0,1\}^n} p_{\vec x}\, \mc N_{\vec x}(\rho_{\vec x}) \right)
\end{equation}
has an optimal choice of input states $\{\rho_{\vec x}\}$ for which
\[
\rho_{\vec x}=\bigotimes_{i=1}^n \tau_{x_i}^{(i)},
\qquad\text{where}\qquad
\tau_0^{(i)}=|\psi_{\max}\rangle\langle\psi_{\max}|_{A_0^{(i)}},\quad
\tau_1^{(i)}=|\phi_{\max}\rangle\langle\phi_{\max}|_{A_1^{(i)}},
\]
where $\ket{\psi_{\max}}$ (resp.\ $\ket{\phi_{\max}}$) is any unit vector in the maximal-eigenspace of $\sigma_0$ (resp.\ $\sigma_1$), and $|\vec x|_0:=|\{i:\ x_i=0\}|$. Equivalently, each freely chosen spin is aligned with a maximal eigenvector of the corresponding fixed state.
\end{conjecture}
Note that since the map $(\rho_{\vec x})_{\vec x} \mapsto \sum_{\vec x\in\{0,1\}^n} p_{\vec x}\, \mc N_{\vec x}(\rho_{\vec x})$ imposes a convex combination and $S(\cdot)$ is concave, there exists an optimal choice with every $\rho_{\vec x}$ pure. Moreover, if $\sigma_0 = \sigma_1$, it can be recovered by the spin alignment conjecture proposed in \cite{Leditzky_2023}. 

Here we provide the rigorous justifications for several special cases. 
\subsubsection{$n = 1$}
Given $p_0,p_1\ge 0$ with $p_0+p_1=1$, consider the entropy minimization problem
\begin{equation}\label{eq:spin-align-1copy}
  \min_{\rho_{A_0}\in\mc D(A_0),\,\rho_{A_1}\in\mc D(A_1)}
  \ S\big(p_0\,\mc N_0(\rho_{A_0})+p_1\,\mc N_1(\rho_{A_1})\big).
\end{equation}
\begin{prop}\label{prop:spin alignment 1 copy}
    The minimization problem~\eqref{eq:spin-align-1copy} has an optimal choice 
    \begin{equation}
  \rho_{A_0}=|\psi_{\max}\rangle\langle \psi_{\max}|,\quad
  \rho_{A_1}=|\phi_{\max} \rangle \langle \phi_{\max}| .
\end{equation}
\end{prop}
To prove the above result, we first briefly review the majorization of matrices. For a Hermitian matrix $X$, let $$\lambda(X)=(\lambda_1^\downarrow(X),\dots,\lambda_d^\downarrow(X))$$ be the vector of eigenvalues in non-increasing order.
For two vectors $x,y\in\mb R^d$, we say $x$ is \emph{majorized} by $y$, written $x\prec y$, if
\[
\sum_{i=1}^k x_i^\downarrow \le \sum_{i=1}^k y_i^\downarrow \quad (k=1,\dots,d-1),
\qquad
\sum_{i=1}^d x_i = \sum_{i=1}^d y_i.
\]
A function $f$ on probability vectors is \emph{Schur concave} if $x\prec y$ implies $f(x)\ge f(y)$.
The Shannon entropy $H(x)=-\sum_i x_i\log x_i$ is Schur concave; hence the von Neumann entropy
$S(\rho)=H(\lambda(\rho))$ is Schur concave in the spectrum.

For a positive semidefinite matrix $X$, define $X^\downarrow$ to be the diagonal matrix (in some fixed reference basis) whose diagonal entries are the eigenvalues of $X$ written in non-increasing order.
 
We will use the following known tensor rearrangement majorization principle, proved in~\cite{Alhejji_2025}:
\begin{lemma}[Tensor rearrangement majorization]\label{lem:tensor-rearrangement}
Let $B_1,B_2\succeq 0$ act on $\mb C^{d_B}$ and $C_1,C_2\succeq 0$ act on $\mb C^{d_C}$.
Then
\begin{equation}\label{eq:tensor-rearrangement}
\lambda\big(B_1\otimes C_1 + B_2\otimes C_2\big)
\ \prec\
\lambda\big(B_1^\downarrow\otimes C_1^\downarrow + B_2^\downarrow\otimes C_2^\downarrow\big).
\end{equation}
\end{lemma}
\begin{proof}[Proof of Proposition~\ref{prop:spin alignment 1 copy}]
Fix arbitrary unit vectors $|\psi\rangle$ and $|\phi\rangle$ and consider
\[
\Omega:=\Omega(|\psi\rangle\langle\psi|,|\phi\rangle\langle\phi|)
=
p_0\,|\psi\rangle\langle\psi|\otimes \sigma_1
+
p_1\,\sigma_0\otimes |\phi\rangle\langle\phi|.
\]
Apply Lemma~\ref{lem:tensor-rearrangement} with
\[
B_1=p_0\,|\psi\rangle\langle\psi|,
\quad
B_2=p_1\,\sigma_0,
\quad
C_1=\sigma_1,
\quad
C_2=|\phi\rangle\langle\phi|.
\]
We obtain the spectral majorization
\begin{equation}\label{eq:majorization-step}
\lambda(\Omega)
\prec
\lambda(\Omega^\downarrow),
\qquad
\Omega^\downarrow:=
p_0\,(|\psi\rangle\langle\psi|)^\downarrow\otimes \sigma_1^\downarrow
+
p_1\,\sigma_0^\downarrow\otimes (|\phi\rangle\langle\phi|)^\downarrow.
\end{equation}
Since $\Omega$ and $\Omega^\downarrow$ are density operators, their spectra are probability vectors.
Because the von Neumann entropy is Schur concave in the spectrum, \eqref{eq:majorization-step} implies
\begin{equation}\label{eq:entropy-lower}
S(\Omega)=H(\lambda(\Omega)) \ \ge\ H(\lambda(\Omega^\downarrow))=S(\Omega^\downarrow).
\end{equation}
By definition, $\sigma_0^\downarrow$ and $\sigma_1^\downarrow$ are diagonal with eigenvalues in non-increasing order; thus
the first basis vectors correspond to the \emph{maximal} eigenvalues of $\sigma_0$ and $\sigma_1$, respectively.
Therefore, in the eigenbases of $\sigma_0$ and $\sigma_1$, we can rewrite $\Omega^\downarrow$ as
\[
\Omega^\downarrow
=
p_0\,|\psi_{\max}\rangle\langle\psi_{\max}|\otimes \sigma_1
+
p_1\,\sigma_0\otimes |\phi_{\max}\rangle\langle\phi_{\max}|,
\]
i.e.\ $\Omega^\downarrow=\Omega(|\psi_{\max}\rangle\langle\psi_{\max}|,|\phi_{\max}\rangle\langle\phi_{\max}|)$ up to unitary conjugation,
and hence has the same entropy as that aligned-output state.

Combining with \eqref{eq:entropy-lower}, we conclude that for every pure pair $(|\psi\rangle,|\phi\rangle)$,
\[
S\big(\Omega(|\psi\rangle\langle\psi|,|\phi\rangle\langle\phi|)\big)
\ \ge\
S\big(\Omega(|\psi_{\max}\rangle\langle\psi_{\max}|,|\phi_{\max}\rangle\langle\phi_{\max}|)\big).
\]
Since an optimizer exists among pure pairs, this shows that \eqref{eq:spin-align-1copy}
admits an optimizer with $\rho_{A_0}=|\psi_{\max}\rangle\langle\psi_{\max}|$ and
$\rho_{A_1}=|\phi_{\max}\rangle\langle\phi_{\max}|$, as claimed.
\end{proof}

\subsubsection{R\'enyi-2 entropy}
Recall that the Rényi-2 entropy is defined by 
\begin{equation}
    S_2(\rho) := -\log \Tr(\rho^2).
\end{equation}
Thus minimizing $S_2(\rho)$ is equivalent to maximizing $\Tr(\rho^2)$. Dealing with $2$-norm can greatly simplify the optimization. In fact, one can separate the mixture and optimize each term, while for von Neumann entropy, the non-linear nature obstructs this analysis. We have the following:
\begin{prop}[Spin alignment for R\'enyi-2 entropy]\label{prop:renyi2-spin-alignment}
Fix a probability distribution $\{p_{\vec x}\}_{\vec x\in\{0,1\}^n}$, the minimization of the R\'enyi-2 entropy
\begin{equation}\label{eqn:spin Renyi}
  \min_{\rho_{\vec x}} S_2 \left(\sum_{\vec x\in\{0,1\}^n} p_{\vec x}\, \mc N_{\vec x}(\rho_{\vec x}) \right)
\end{equation}
has an optimal choice of input states $\{\rho_{\vec x}\}$ for which
\[
\rho_{\vec x}=\bigotimes_{i=1}^n \tau_{x_i}^{(i)},
\qquad\text{where}\qquad
\tau_0^{(i)}=|\psi_{\max}\rangle\langle\psi_{\max}|_{A_0^{(i)}},\quad
\tau_1^{(i)}=|\phi_{\max}\rangle\langle\phi_{\max}|_{A_1^{(i)}},
\]
and $|\psi_{\max}\rangle$ (resp.\ $|\phi_{\max}\rangle$) is any eigenvector of $\sigma_0$ (resp.\ $\sigma_1$)
corresponding to its largest eigenvalue $\lambda_0$ (resp.\ $\lambda_1$).
\end{prop}
\begin{proof}
Define the constants
\begin{equation}
    \lambda_0:=\lambda_{\max}(\sigma_0),\qquad \lambda_1:=\lambda_{\max}(\sigma_1),\qquad
\alpha_0:=\Tr(\sigma_0^2),\qquad \alpha_1:=\Tr(\sigma_1^2).
\end{equation}
For $\vec x,\vec y\in\{0,1\}^n$ define the counts
\begin{equation}
    N_{ab}(\vec x,\vec y):=\bigl|\{i\in[n] : (x_i,y_i)=(a,b)\}\bigr|,\qquad a,b\in\{0,1\}.
\end{equation}
We show that the maximizer of \begin{align*}
\Tr\big[\big(\sum_{\vec x\in\{0,1\}^n} p_{\vec x}\, \mc N_{\vec x}(\rho_{\vec x}) \big)^2 \big]
=
\sum_{\vec x,\vec y} p_{\vec x}p_{\vec y}\,
\Tr\big[\mathcal N_{\vec x}(\rho_{\vec x})\,\mathcal N_{\vec y}(\rho_{\vec y})\big].
\end{align*}
can be chosen as $$\rho_{\vec x}=\bigotimes_{i=1}^n \tau_{x_i}^{(i)},
\qquad\text{where}\qquad
\tau_0^{(i)}=|\psi_{\max}\rangle\langle\psi_{\max}|_{A_0^{(i)}},\quad
\tau_1^{(i)}=|\phi_{\max}\rangle\langle\phi_{\max}|_{A_1^{(i)}}.$$
In fact, in Appendix~\ref{app:upper bound}, we show that given $\{\rho_{\vec x}\}_{\vec x \in \{0,1\}^n}$, for any $\vec x, \vec y \in \{0,1\}^n$, we have
\begin{align}\label{eqn:key upper bound Renyi}
    \Tr\big[\mathcal N_{\vec x}(\rho_{\vec x})\,\mathcal N_{\vec y}(\rho_{\vec y})\big] \le \alpha_1^{N_{00}(\vec x,\vec y)}\alpha_0^{N_{11}(\vec x,\vec y)}
(\lambda_0\lambda_1)^{N_{01}(\vec x,\vec y)+N_{10}(\vec x,\vec y)}.
\end{align}
On the other hand, via direct calculation, 
\begin{align*}
    \Tr\!\big[\mathcal N_{x_i}(\tau_{x_i}^{(i)})\,\mathcal N_{y_i}(\tau_{y_i}^{(i)})\big]
=
\begin{cases}
\Tr(\sigma_1^2)=\alpha_1,& (x_i,y_i)=(0,0),\\[2pt]
\Tr(\sigma_0^2)=\alpha_0,& (x_i,y_i)=(1,1),\\[2pt]
\Tr(|\psi_{\max}\rangle\langle\psi_{\max}|\sigma_0)\,
\Tr(\sigma_1|\phi_{\max}\rangle\langle\phi_{\max}|)=\lambda_0\lambda_1,& (x_i,y_i)=(0,1),\\[2pt]
\Tr(\sigma_0|\psi_{\max}\rangle\langle\psi_{\max}|)\,
\Tr(|\phi_{\max}\rangle\langle\phi_{\max}|\sigma_1)=\lambda_0\lambda_1,& (x_i,y_i)=(1,0).
\end{cases}
\end{align*}
which concludes the proof that $\rho_{\vec x}=\bigotimes_{i=1}^n \tau_{x_i}^{(i)}$ is an optimal choice of~\eqref{eqn:spin Renyi}
\end{proof}

In the next section, we will illustrate how to derive single-letter expression for quantum capacities and one-way distillable entanglement using Spin alignment phenomenon.


\section{Explicit examples of non-degradable states with single-letter one-way
distillable entanglement}\label{sec:examples}
In this section, we provide several examples of non-degradable states with single-letter distillable entanglement using the framework proposed in the previous section. For the first class consisting of states with weaker degradability, one can choose the Choi state of quantum channels proposed in \cite{Smith_2025}.
Therefore, we focus on the second and third class in the remaining section.


\subsection{Mixture of degradable and anti-degradable states with orthogonal support}

To illustrate our general framework, we consider the state $\rho_{AB}(s)$ given by 
\begin{align*}
 \frac{1}{3}\Big(
|0\rangle\!\langle0|_A\otimes\big[s\,|0\rangle\!\langle0|_B+(1-s)\,|2\rangle\!\langle2|_B\big]
+|0\rangle\!\langle1|_A\otimes\sqrt{s}\,|0\rangle\!\langle2|_B
+|1\rangle\!\langle0|_A\otimes\sqrt{s}\,|2\rangle\!\langle0|_B
+|1\rangle\!\langle1|_A\otimes|2\rangle\!\langle2|_B
+|2\rangle\!\langle2|_A\otimes|2\rangle\!\langle2|_B
\Big).
\end{align*}
This is an example of a state with an useless component having orthogonal support on Alice's side to the useful component. The antidegradable state here is a separable state $\ketbra{2}{2}_A\otimes\ketbra{2}{2}_B$ (appears in the mixture with probability $1/3$) and the useful component here is the Choi state of an amplitude damping channel from the $\rm{Span}\{\ket{0},\ket{1}\}\subset \mathcal{H}_A$ to $\rm{Span}\{\ket{1},\ket{2}\}\subset \mathcal{H}_B$, parameterized by $s$. Therefore, from Theorem \ref{main:class 1}, we conclude that this state has a single-letter distillable entanglement, given by
\begin{align}
     D_{\to}(\rho_{AB}(s))  = D^{(1)}_{\to}(\rho_{AB}(s)) 
      = \frac{2}{3} \max\left\{h\left(\frac{s}{2}\right) - h\left(\frac{1+s}{2}\right), 0 \right\}.
\end{align}


\subsection{Flagged mixture of degradable and antidegradable states}
A concrete and very explicit family comes from \emph{flagged} mixtures of amplitude damping channels
(from \cite{Smith_2025}).
Let ${\rm AD}_\gamma$ be the qubit amplitude damping channel with Kraus operators
\[
E_0 = |0\rangle\!\langle 0| + \sqrt{1-\gamma}\,|1\rangle\!\langle 1|,
\qquad
E_1 = \sqrt{\gamma}\,|0\rangle\!\langle 1| ,
\qquad \gamma\in[0,1].
\]
Fix parameters $\gamma_0,\gamma_1\in[0,1]$ and a mixing probability $p\in[0,1]$, and define the flagged channel
\[
\mathcal{N}^{A\to BF}(\rho)
:= p\,{\rm AD}_{\gamma_0}(\rho)\otimes |0\rangle\!\langle 0|_F
 + (1-p)\,{\rm AD}_{\gamma_1}(\rho)\otimes |1\rangle\!\langle 1|_F ,
\]
where the classical flag $F$ is available to Bob.  Its (normalized) Choi state
$\rho_{RBF}^{\mathcal{N}} := (\mathrm{id}_R\otimes \mathcal{N})(|\Phi^+\rangle\!\langle\Phi^+|)$,
with $|\Phi^+\rangle = (|00\rangle+|11\rangle)/\sqrt{2}$, is
\begin{equation}\label{eq:choi-flagged-AD}
\rho_{RBF}^{\mathcal{N}}
= p\,\rho_{RB}^{(\gamma_0)}\otimes |0\rangle\!\langle 0|_F
 + (1-p)\,\rho_{RB}^{(\gamma_1)}\otimes |1\rangle\!\langle 1|_F,
\end{equation}
where, in the computational basis $\{|00\rangle,|01\rangle,|10\rangle,|11\rangle\}$ of $RB$,
\begin{equation}\label{eq:choi-AD-explicit}
\rho_{RB}^{(\gamma)}
=
\frac{1}{2}\begin{pmatrix}
1 & 0 & 0 & \sqrt{1-\gamma}\\
0 & 0 & 0 & 0\\
0 & 0 & \gamma & 0\\
\sqrt{1-\gamma} & 0 & 0 & 1-\gamma
\end{pmatrix}.
\end{equation}
This flagged family for some parameter regions provides explicit \emph{non-degradable} Choi states for which the weaker
information-dominance property holds and hence $D_\to(\rho_{RBF}^{\mathcal{N}})$ is single-letter, see~\cite{Smith_2025}.


\subsection{Generalized direct sum channels and their corresponding states}
In this subsection, motivated by~\cite{wu2025}, we present a class of quantum channels and states with single-letter expression for capacities and one-way distillable entanglement. We first show that this class of channels can have single-letter quantum capacity using the spin alignment phenomenon.

We say a completely positive map $\Phi$ has a generalized (partially coherent) direct sum structure if 
\begin{align*}
    \Phi \begin{pmatrix}
        X_{00} & X_{01} \\
        X_{10} & X_{11}
    \end{pmatrix} =  \begin{pmatrix}
        \Phi_{0}(X_{00}) & \Phi_{01}(X_{01}) \\
        \Phi_{10}(X_{10}) & \Phi_{1}(X_{11})
    \end{pmatrix}
\end{align*}
where $\Phi_0,\Phi_1$ are completely positive. It is shown in \cite{Chessa_2021} that $\Phi$ admits partially coherent direct sum structure if and only if the Kraus operators of $\Phi$ are block diagonal, i.e., $\Phi(X) = \sum_k E_k X E_k^\dagger$ with 
\begin{equation}
    E_k = \begin{pmatrix}
        E_k^0 & 0 \\
        0 & E_k^1
    \end{pmatrix}.
\end{equation}
A key observation is that the optimal states achieving the maximal coherent information of $\Phi^{\otimes n}$ can be chosen to be block diagonal.
\begin{lemma}\label{lem:block-diagonal-optimizer}
    Suppose $\Phi$ is a generalized direct sum channel with Kraus operators $\begin{pmatrix}
        E_k^0 & 0 \\
        0 & E_k^1
    \end{pmatrix}$. Then
    \begin{align*}
        \mc Q^{(1)}(\Phi^{\otimes n})
        = \max\Big\{I_c(\rho_n,\Phi^{\otimes n}) :\ 
        \rho_n = \sum_{\vec x \in \{0,1\}^n} p_{\vec x}\ |\vec x\rangle \langle \vec x | \otimes \rho_{\vec x},\ 
        \rho_{\vec x} \in \mc D_{d_0^{|\vec x|} d_1^{n-|\vec x|}},\ 
        \sum_{\vec x \in \{0,1\}^n} p_{\vec x}=1\Big\}.
    \end{align*}
\end{lemma}

\begin{proof}
    Using the direct sum structure of the Hilbert space, any $\rho_n \in \mc D_{(d_0+d_1)^n}$ can be decomposed as
    \begin{equation*}
        \rho_n = \sum_{\vec x,\vec y \in \{0,1\}^n} |\vec x\rangle \langle \vec y | \otimes \rho_{\vec x,\vec y},\quad
        \rho_{\vec x,\vec y} \in \mb M_{d_0^{|\vec x|} d_1^{n-|\vec x|} \times d_0^{|\vec y|} d_1^{n-|\vec y|}}.
    \end{equation*}
    Define the pinching map (conditional expectation onto block-diagonal entries)
    \begin{equation*}
        \mc P(\rho_n)
        := \sum_{\vec x \in \{0,1\}^n} (|\vec x\rangle\langle \vec x|\otimes \mb I)\ \rho_n\ (|\vec x\rangle\langle \vec x|\otimes \mb I)
        = \sum_{\vec x} |\vec x\rangle\langle \vec x|\otimes \rho_{\vec x,\vec x}.
    \end{equation*}

    \textbf{Step 1: the complementary output is unchanged by pinching.}
    Since each Kraus operator of $\Phi$ is block diagonal, every Kraus operator of $\Phi^{\otimes n}$ is also block diagonal with respect to the decomposition indexed by $\vec x\in\{0,1\}^n$. Hence, for $\vec x\neq \vec y$ one has
    \begin{equation*}
        \Tr\!\big(E_{\vec k}\ (|\vec x\rangle\langle \vec y|\otimes \rho_{\vec x,\vec y})\ E_{\vec \ell}^\dagger\big)=0,
    \end{equation*}
    and therefore the complementary channel depends only on the diagonal blocks:
    \begin{equation*}
        (\Phi^c)^{\otimes n}(\rho_n) = (\Phi^c)^{\otimes n}(\mc P(\rho_n)).
    \end{equation*}

    \textbf{Step 2: the main output entropy increases under pinching.}
    Set $\omega := \Phi^{\otimes n}(\rho_n)$ and $\omega' := \Phi^{\otimes n}(\mc P(\rho_n))$.
    By construction, $\omega'$ is obtained from $\omega$ by pinching with respect to the orthogonal projections $\{|\vec x\rangle\langle \vec x|\otimes \mb I\}_{\vec x}$ on the output space. Equivalently, $\omega'$ is a convex combination of unitary conjugations of $\omega$ (hence the spectrum of $\omega'$ majorizes that of $\omega$), and since the von Neumann entropy is Schur concave we get
    \begin{equation*}
        S(\Phi^{\otimes n}(\rho_n)) = S(\omega) \le S(\omega') = S(\Phi^{\otimes n}(\mc P(\rho_n))).
    \end{equation*}
    This is exactly the majorization/pinching trick~\cite{bhatia2013matrix}.

    Combining Step~1 and Step~2,
    \begin{equation*}
        I_c(\rho_n,\Phi^{\otimes n})
        = S(\Phi^{\otimes n}(\rho_n)) - S((\Phi^c)^{\otimes n}(\rho_n))
        \le S(\Phi^{\otimes n}(\mc P(\rho_n))) - S((\Phi^c)^{\otimes n}(\mc P(\rho_n)))
        = I_c(\mc P(\rho_n),\Phi^{\otimes n}),
    \end{equation*}
    so an optimizer exists among block-diagonal states of the claimed form.
\end{proof}

Now we show the following using the spin alignment phenomenon, which generalizes the result in~\cite{wu2025quantum} to different dimensions.

\begin{prop}\label{prop:GDS-single-letter}
    Assume $\Phi_0(X)=\Tr(X)\mb I_{d_0'}/d_0'$ and $\Phi_1(X)=\Tr(X)\mb I_{d_1'}/d_1'$, and we choose the generalized direct sum channel as
    \begin{equation}\label{eq:Phi-special}
        \Phi \begin{pmatrix}
        X_{00} & X_{01} \\
        X_{10} & X_{11}
    \end{pmatrix}
    = \begin{pmatrix}
        \Tr(X_{00})\frac{\mb I_{d_0'}}{d_0'} & X_{01}^T \\
        X_{10}^T & \Tr(X_{11})\frac{\mb I_{d_1'}}{d_1'}
    \end{pmatrix}.
    \end{equation}
    Then $\mc Q(\Phi)=\mc Q^{(1)}(\Phi)$.
\end{prop}

\begin{proof}
    \textbf{Step 0: enlarge dimensions without changing coherent information.}
    When $d_0',d_1'$ do not match the sizes needed to place $X_{01}^T$ as an off-diagonal block, we work with an enlarged GDS channel
    $\wt\Phi:\mc B(\mb C^{d_0}\oplus \mb C^{d_1})\to \mc B(\mb C^{D_0}\oplus \mb C^{D_1})$ where
    \begin{equation}\label{eq:D0D1}
        D_0=\max\{d_0',d_1\},\qquad D_1=\max\{d_0,d_1'\}.
    \end{equation}
    Define $\wt\Phi$ exactly as in your construction, i.e. it has the block form
    \begin{equation*}
        \wt \Phi \begin{pmatrix}
            X_{00} & X_{01} \\
            X_{10} & X_{11}
        \end{pmatrix}
        = \begin{pmatrix}
            \Tr(X_{00})\frac{\mb I_{d_0'}}{d_0'} & 0 & X_{01}^T & 0 \\
            0 & 0 & 0 & 0 \\
            X_{10}^T & 0 & \Tr(X_{11})\frac{\mb I_{d_1'}}{d_1'} & 0 \\
            0 & 0 & 0 & 0
        \end{pmatrix},
    \end{equation*}
    so that $\Phi$ is obtained from $\wt\Phi$ by restricting to the supported output subspace (equivalently, composing with an isometry and its adjoint).
    Therefore, for every $n$ and every input state $\rho_n$,
    \begin{equation*}
        S(\Phi^{\otimes n}(\rho_n)) = S(\wt\Phi^{\otimes n}(\rho_n)),\qquad
        S((\Phi^c)^{\otimes n}(\rho_n)) = S((\wt\Phi^c)^{\otimes n}(\rho_n)),
    \end{equation*}
    and hence $\mc Q^{(1)}(\Phi^{\otimes n})=\mc Q^{(1)}(\wt\Phi^{\otimes n})$. In particular, it suffices to prove the claim for $\wt\Phi$.

    \textbf{Step 1: identify the spin-alignment structure on the complementary channel.}
    With Kraus operators as in your definition, $\wt\Phi^c$ can be chosen so that
    \begin{equation}\label{eq:Phi-tilde-complement}
        \wt \Phi^c \begin{pmatrix}
            X_{00} & X_{01} \\
            X_{10} & X_{11}
        \end{pmatrix}
        = \frac{1}{d_0'}P_{d_0'} \otimes \widehat X_{00} \;+\; \widehat X_{11} \otimes \frac{1}{d_1'}P_{d_1'} ,
    \end{equation}
    where $\widehat X_{00}$ (resp.\ $\widehat X_{11}$) is the embedding of $X_{00}$ (resp.\ $X_{11}$) into $\mc B(\mb C^{D_1})$ (resp.\ $\mc B(\mb C^{D_0})$), and $P_{d_0'}$ (resp.\ $P_{d_1'}$) is the projection onto $\mb C^{d_0'}\subset \mb C^{D_0}$ (resp.\ $\mb C^{d_1'}\subset \mb C^{D_1}$).

    In particular, for any block-diagonal input state
    \begin{equation*}
        \rho = \begin{pmatrix}
            \rho_{A_0} & 0 \\
            0 & \rho_{A_1}
        \end{pmatrix},\qquad p_0:=\Tr(\rho_{A_0}),\ p_1:=\Tr(\rho_{A_1}),
    \end{equation*}
    the complementary output is exactly a two-branch mixture of the spin-alignment form:
    \begin{equation}\label{eq:spin-align-instance}
        \wt\Phi^c(\rho)
        = p_0\Big(\frac{1}{d_0'}P_{d_0'}\Big)\otimes \widehat{\rho}_{A_0}
        \;+\;
        p_1\ \widehat{\rho}_{A_1}\otimes \Big(\frac{1}{d_1'}P_{d_1'}\Big).
    \end{equation}

    \textbf{Step 2: reduce the $n$-copy optimization to the aligned two-dimensional subspace (Spin Alignment Conjecture).}
    By Lemma~\ref{lem:block-diagonal-optimizer}, for every $n$ we may restrict to inputs of the form
    \begin{equation*}
        \rho_n = \sum_{\vec x \in \{0,1\}^n} p_{\vec x}\ |\vec x\rangle\langle \vec x|\otimes \rho_{\vec x}.
    \end{equation*}
    For our specific choice~\eqref{eq:Phi-special} (hence also for $\wt\Phi$), the main output $\wt\Phi^{\otimes n}(\rho_n)$ depends on $\rho_{\vec x}$ only through the weights $p_{\vec x}$, because each diagonal branch is completely depolarizing. Therefore, for fixed $\{p_{\vec x}\}$, maximizing $I_c(\rho_n,\wt\Phi^{\otimes n})$ is equivalent to minimizing the entropy of the complementary output
    \begin{equation*}
        (\wt\Phi^c)^{\otimes n}(\rho_n)
        = \sum_{\vec x \in \{0,1\}^n} p_{\vec x}\ (\wt\Phi^c)^{\otimes n}\big(|\vec x\rangle\langle \vec x|\otimes \rho_{\vec x}\big),
    \end{equation*}
    and each summand has the spin-alignment structure inherited from~\eqref{eq:spin-align-instance} on every tensor factor.

    Invoking the (multi-copy) Spin Alignment Conjecture for the entropy-minimization problem associated with~\eqref{eq:spin-align-instance}, we may choose an optimizer such that, for each $\vec x$, the state $\rho_{\vec x}$ is supported on the tensor product of one-dimensional subspaces spanned by a fixed unit vector in $\mb C^{d_0}$ and a fixed unit vector in $\mb C^{d_1}$ (i.e.\ the ``aligned spins''). Equivalently, there exists an optimizer supported on
    \begin{equation}\label{eq:aligned-subspace}
        \Big(\mathrm{span}\{\ket{0},\ket{d_1}\}\Big)^{\otimes n}\subset (\mb C^{d_0}\oplus \mb C^{d_1})^{\otimes n},
    \end{equation}
    where $\ket{0}$ denotes a unit vector in the first summand $\mb C^{d_0}$ and $\ket{d_1}$ denotes a unit vector in the second summand $\mb C^{d_1}$ (after fixing an orthonormal basis).
    Consequently, for every $n$,
    \begin{equation}\label{eq:reduce-to-subchannel}
        \mc Q^{(1)}(\wt\Phi^{\otimes n})=\mc Q^{(1)}(\Psi^{\otimes n}),
    \end{equation}
    where $\Psi$ is the restriction of $\wt\Phi$ to the two-dimensional input subspace $\mathrm{span}\{\ket{0},\ket{d_1}\}$.

    \textbf{Step 3: the restricted channel $\Psi$ is degradable.}
    On $\mathrm{span}\{\ket{0},\ket{d_1}\}$, both diagonal blocks $X_{00}$ and $X_{11}$ are scalars, hence by~\eqref{eq:Phi-tilde-complement} the complementary output $\Psi^c(\cdot)$ depends only on the block traces (and ignores the off-diagonal entry).
    Let $P_0$ and $P_1$ be the orthogonal projections onto the two output summands of $\wt\Phi$.
    Define a CPTP map $\mc D$ on the output space of $\Psi$ by
    \begin{equation}\label{eq:degrading-map}
        \mc D(Y)
        := \Tr(P_0 Y)\Big(\frac{1}{d_0'}P_{d_0'}\Big)\otimes \ket{0}\bra{0}
        \;+\;
        \Tr(P_1 Y)\ \ket{d_1}\bra{d_1}\otimes \Big(\frac{1}{d_1'}P_{d_1'}\Big).
    \end{equation}
    Since $\Tr(P_0\Psi(\rho))=\Tr(\rho_{A_0})$ and $\Tr(P_1\Psi(\rho))=\Tr(\rho_{A_1})$, and since $\Psi^c$ depends only on these two weights on the restricted subspace, one checks directly from~\eqref{eq:Phi-tilde-complement} that
    \begin{equation*}
        \Psi^c = \mc D \circ \Psi,
    \end{equation*}
    i.e.\ $\Psi$ is degradable. Therefore $\mc Q(\Psi)=\mc Q^{(1)}(\Psi)$.

    \textbf{Step 4: conclude single-letter capacity for $\Phi$.}
    Using~\eqref{eq:reduce-to-subchannel} and degradability of $\Psi$,
    \begin{equation*}
        \mc Q(\Phi)
        = \lim_{n\to\infty}\frac{1}{n}\mc Q^{(1)}(\Phi^{\otimes n})
        = \lim_{n\to\infty}\frac{1}{n}\mc Q^{(1)}(\wt\Phi^{\otimes n})
        = \lim_{n\to\infty}\frac{1}{n}\mc Q^{(1)}(\Psi^{\otimes n})
        = \mc Q^{(1)}(\Psi)
        = \mc Q^{(1)}(\wt\Phi)
        = \mc Q^{(1)}(\Phi),
    \end{equation*}
    which proves $\mc Q(\Phi)=\mc Q^{(1)}(\Phi)$.
\end{proof}
Now we provide the proof of the single-letterization of one-way distillable entanglement for a class of generalised direct sum (GDS) states as an example case of the Spin alignment phenomena.
\begin{prop}\label{thm:GDS-state-additivity}
Let $A,B \cong \mathbb{C}^{d_0+d_1}$ and define
\begin{align*}
\rho_{AB}
=
\frac{1}{2d_0 d_1}
\sum_{i=0}^{d_0-1}
\sum_{j=0}^{d_1-1}
\left(
\ket{i}\ket{j}
+
\ket{j+d_0}\ket{i+d_1}
\right)
\left(
\bra{i}\bra{j}
+
\bra{j+d_0}\bra{i+d_1}
\right).
\end{align*}
Then $\rho_{AB}$ has single-letter one-way distillable entanglement:
\begin{align*}
D_{\to}(\rho_{AB}) = D^{(1)}_{\to}(\rho_{AB}).
\end{align*}
\end{prop}

\begin{proof}
We prove that for every $n\in\mathbb{N}$,
\begin{align*}
D^{(1)}_{\to}(\rho_{AB}^{\otimes n})
=
n\,D^{(1)}_{\to}(\rho_{AB}),
\end{align*}
which implies the claim by regularization.

\medskip
\noindent
\textbf{Step 0: a canonical instrument.}
For $0\le i<d_0$ and $0\le j<d_1$ define the (orthonormal) vectors
\begin{align*}
\ket{\psi_{ij}}
=
\frac{1}{\sqrt{2}}\Bigl(\ket{i}\ket{j}+\ket{j+d_0}\ket{i+d_1}\Bigr),
\end{align*}
so that $\rho_{AB}=\frac{1}{d_0d_1}\sum_{i,j}\ketbra{\psi_{ij}}{\psi_{ij}}$ is the maximally mixed state on the span of $\{\ket{\psi_{ij}}\}_{i,j}$.
Define an instrument $\mathcal{T}:A\to AM$ with Kraus operators
\begin{align*}
K_{ij}
=
\frac{1}{\sqrt{d_1}}\ketbra{i}{i}
+
\frac{1}{\sqrt{d_0}}\ketbra{d_0+j}{d_0+j},
\qquad
0\le i<d_0,\; 0\le j<d_1.
\end{align*}
A direct computation gives $\sum_{i,j}K_{ij}^\dagger K_{ij}=\mathbb{I}_A$, hence $\mathcal{T}$ is trace-preserving and $\mathcal{T}^{\otimes n}$ is a valid candidate for $D^{(1)}_{\to}(\rho_{AB}^{\otimes n})$.

\medskip
\noindent
\textbf{Step 1: reduction to block-diagonal POVM elements (pinching).}
Decompose $A=\mathbb{C}^{d_0}\oplus\mathbb{C}^{d_1}$ and define projectors
\begin{align*}
\Pi_0=\sum_{i=0}^{d_0-1}\ketbra{i}{i},
\qquad
\Pi_1=\sum_{j=0}^{d_1-1}\ketbra{j+d_0}{j+d_0}.
\end{align*}
For $\vec{x}\in\{0,1\}^n$, set
\begin{align*}
\Pi_{\vec{x}}=\Pi_{x_1}\otimes\cdots\otimes\Pi_{x_n}.
\end{align*}

Let $\mathcal{T}_1:A^{\otimes n}\to A^{\otimes n}M$ be any trace-preserving instrument with Kraus operators $\{L_m\}_m$ and POVM elements $E_m=L_m^\dagger L_m$.
Consider the pinched POVM elements
\begin{align*}
\widetilde E_m
:=
\sum_{\vec{x}\in\{0,1\}^n}\Pi_{\vec{x}}E_m\Pi_{\vec{x}}.
\end{align*}
Since $\sum_m E_m=\mathbb{I}$, we still have $\sum_m \widetilde E_m=\mathbb{I}$, so $\{\widetilde E_m\}_m$ comes from a trace-preserving instrument $\widetilde{\mathcal{T}}_1$.

The GDS structure of $\rho_{AB}$ implies that, conditioned on outcome $m$, the joint state of $B^{\otimes n}$ and of a purification environment depends only on the diagonal blocks $\Pi_{\vec{x}}E_m\Pi_{\vec{x}}$ and is insensitive to the off-diagonal blocks $\Pi_{\vec{x}}E_m\Pi_{\vec{y}}$ with $\vec{x}\neq \vec{y}$.
Moreover, on the Bob system the replacement $E_m\mapsto \widetilde E_m$ corresponds to pinching the post-measurement state in the direct-sum decomposition indexed by $\vec{x}$, which increases the entropy by the Ky Fan majorization trick~\cite{bhatia2013matrix}.
Since $I(A^{\otimes n}\rangle B^{\otimes n}M)=S(B^{\otimes n}M)-S(E^{\otimes n}M)$ for a purification, and the environment term is unchanged under the above replacement, we obtain
\begin{align*}
I(A^{\otimes n}\rangle B^{\otimes n}M)_{\widetilde{\mathcal{T}}_1(\rho^{\otimes n})}
\ge
I(A^{\otimes n}\rangle B^{\otimes n}M)_{\mathcal{T}_1(\rho^{\otimes n})}.
\end{align*}
It therefore suffices to optimize over block-diagonal instruments.

\medskip
\noindent
\textbf{Step 2: spin alignment within each block.}
Fix a block-diagonal instrument (still denoted $\mathcal{T}_1$) with POVM elements $E_m=\sum_{\vec{x}}\Pi_{\vec{x}}E_m\Pi_{\vec{x}}$.
For each $m$ and $\vec{x}$ define the weights
\begin{align*}
\alpha_{m,\vec{x}}:=\Tr(\Pi_{\vec{x}}E_m\Pi_{\vec{x}})\ge 0.
\end{align*}
For our state $\rho_{AB}^{\otimes n}$, the marginal on $B^{\otimes n}M$ depends on $E_m$ only through the collection $\{\alpha_{m,\vec{x}}\}_{\vec{x}}$ (intuitively, within a fixed block $\vec{x}$ the state is maximally mixed on the corresponding support, so only the total weight matters for Bob's entropy).
Thus, keeping $\{\alpha_{m,\vec{x}}\}$ fixed, to increase coherent information it is enough to \emph{minimize} the entropy contribution coming from the environment.

By the (multi-copy) spin alignment property applied to the entropy-minimization problem induced by $\rho_{AB}^{\otimes n}$, the environment entropy is minimized when each positive block $\Pi_{\vec{x}}E_m\Pi_{\vec{x}}$ is replaced by a rank-one projector with the same weight.
Concretely, define positive elements
\begin{align*}
E'_{m}
=
\sum_{\vec{x}\in\{0,1\}^n}
\alpha_{m,\vec{x}}
\bigotimes_{i=1}^n
\ketbra{t(x_i)}{t(x_i)},
\end{align*}
where $t(0)=0$ and $t(1)=d_0$.
Let $\mathcal{T}_2$ denote the corresponding (possibly non–trace-preserving) instrument (as in Definition~\ref{def: Non-TP instrument}). Then, by construction,
\begin{align*}
I(A^{\otimes n}\rangle B^{\otimes n}M)_{\mathcal{T}_2(\rho^{\otimes n})}
\ge
I(A^{\otimes n}\rangle B^{\otimes n}M)_{\mathcal{T}_1(\rho^{\otimes n})}.
\end{align*}

\medskip
\noindent
\textbf{Step 3: reduction to a degradable tensor-power state.}
Define
\begin{align*}
\sigma_0
=
(d_0 d_1)
(K_{00}\otimes\mathbb{I})
\rho_{AB}
(K_{00}^\dagger\otimes\mathbb{I}),
\qquad
\sigma
=
\sigma_0^{\otimes n}.
\end{align*}
One verifies (see the appendix) that $\sigma_0$ is a degradable state, hence so is $\sigma$.
Moreover, from the explicit form of $E'_m$, one can construct an instrument $\mathcal{T}_3$ acting on $\sigma$ such that the classical-quantum branch structure matches that of $\mathcal{T}_2(\rho^{\otimes n})$ and
\begin{align*}
I(A^{\otimes n}\rangle B^{\otimes n}M)_{\mathcal{T}_2(\rho^{\otimes n})}
=
I(A^{\otimes n}\rangle B^{\otimes n}M)_{\mathcal{T}_3(\sigma)}.
\end{align*}

Since $\sigma$ is degradable, coherent information is maximized by the trivial instrument~\cite{Leditzky_2018}. Therefore
\begin{align*}
I(A^{\otimes n}\rangle B^{\otimes n})_{\sigma}
\ge
I(A^{\otimes n}\rangle B^{\otimes n}M)_{\mathcal{T}_3(\sigma)}.
\end{align*}
Combining inequalities gives the upper bound
\begin{align}
I(A^{\otimes n}\rangle B^{\otimes n})_{\sigma}
\ge
D^{(1)}_{\to}(\rho_{AB}^{\otimes n}).
\end{align}

\medskip
\noindent
\textbf{Step 4: a matching lower bound from $\mathcal{T}^{\otimes n}$.}
For all $x_i\in (0,\dots, d_0-1)$ and $y_i \in (0,\dots,d_1-1)$, the normalized branch states
\begin{align*}
(d_0d_1)^n
\bigotimes_{i=1}^n
(K_{x_i y_i}\otimes\mathbb{I})
\rho_{AB}
(K_{x_i y_i}^\dagger\otimes\mathbb{I})
\end{align*}
are related to $\sigma$ by local unitaries on $A^{\otimes n}$ and $B^{\otimes n}$ (they simply relabel the basis vectors within each direct-sum block). Hence they have the same coherent information as $\sigma$, and these branches are precisely the outcomes of $\mathcal{T}^{\otimes n}$. Therefore
\begin{align}\label{eq:low-bound}
I(A^{\otimes n}\rangle B^{\otimes n})_{\sigma}
=
I(A^{\otimes n}\rangle B^{\otimes n}M)_{\mathcal{T}^{\otimes n}(\rho^{\otimes n})}
\le
D^{(1)}_{\to}(\rho_{AB}^{\otimes n}).
\end{align}
Combining the bounds yields
\begin{align}
D^{(1)}_{\to}(\rho_{AB}^{\otimes n})
=
I(A^{\otimes n}\rangle B^{\otimes n})_{\sigma}.
\end{align}
Since $\sigma=\sigma_0^{\otimes n}$,
\begin{align*}
I(A^{\otimes n}\rangle B^{\otimes n})_{\sigma}
=
n\, I(A\rangle B)_{\sigma_0}.
\end{align*}
Therefore,
\begin{align*}
D^{(1)}_{\to}(\rho_{AB}^{\otimes n})
=
n\,D^{(1)}_{\to}(\rho_{AB}).
\end{align*}
Taking the regularized limit yields
\begin{align}
D_{\to}(\rho_{AB})
=
\lim_{n\to\infty}
\frac{1}{n}
D^{(1)}_{\to}(\rho_{AB}^{\otimes n})
=
D^{(1)}_{\to}(\rho_{AB}).
\end{align}
\end{proof}
 One can notice the structural similarity in the proofs of the Propositions \ref{prop:GDS-single-letter} and \ref{thm:GDS-state-additivity} till step 2. However, in step 3 and step 4 the proofs significantly differ. As demonstrated below, this difference explains why it is harder to prove additivity for $D^{(1)}_\to$ than that for $Q^{(1)}$.

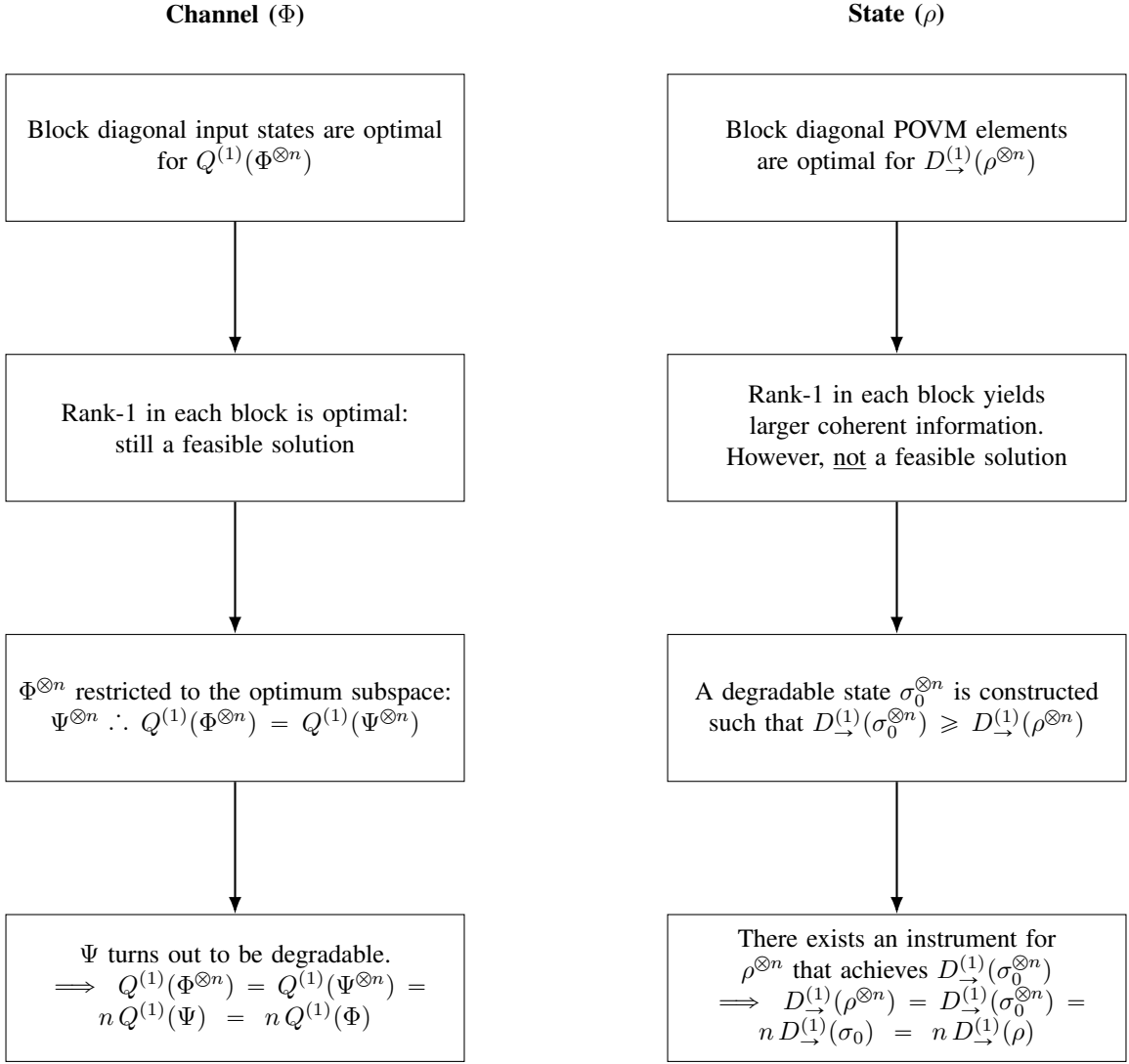
\begin{center}
\begin{tikzpicture}[
    box/.style={
        draw,
        rectangle,
        align=center,
        text width=6cm,
        minimum height=2cm
    },
    step/.style={
        font=\bfseries
    },
    arrow/.style={
        -{Latex},
        thick
    },
    node distance=1.8cm
]

\node at (0,1.8) {\textbf{Channel ($\Phi$)}};
\node at (9,1.8) {\textbf{State ($\rho$)}};

\node[box] (C1) at (0,0)
{Block diagonal input states are optimal\\
for $Q^{(1)}(\Phi^{\otimes n})$};

\node[box, below=of C1] (C2)
{Rank-1 in each block is optimal:\\
still a feasible solution};

\node[box, below=of C2] (C3)
{$\Phi^{\otimes n}$ restricted to the optimum subspace: $\Psi^{\otimes n}$
$\therefore Q^{(1)}(\Phi^{\otimes n}) = Q^{(1)}(\Psi^{\otimes n})$};

\node[box, below=of C3] (C4)
{$\Psi$ turns out to be degradable.\\
$\implies Q^{(1)}(\Phi^{\otimes n}) = Q^{(1)}(\Psi^{\otimes n}) = n\,Q^{(1)}(\Psi)=n\,Q^{(1)}(\Phi)$};

\node[box] (S1) at (9,0)
{Block diagonal POVM elements are optimal for $D^{(1)}_{\to}(\rho^{\otimes n})$};

\node[box, below=of S1] (S2)
{Rank-1 in each block yields larger coherent information.\\
However, \underline{not} a feasible solution};

\node[box, below=of S2] (S3)
{A degradable state $\sigma_0^{\otimes n}$ is constructed such that
$D^{(1)}_{\to}(\sigma_0^{\otimes n}) \ge D^{(1)}_{\to}(\rho^{\otimes n})$};

\node[box, below=of S3] (S4)
{There exists an instrument for $\rho^{\otimes n}$ that achieves $D^{(1)}_{\to}(\sigma_0^{\otimes n})$\\
 $\implies D^{(1)}_{\to}(\rho^{\otimes n}) = D^{(1)}_{\to}(\sigma_0^{\otimes n})= n\,D^{(1)}_{\to}(\sigma_0)= n\,D^{(1)}_{\to}(\rho)$\\
};

\draw[arrow] (C1) -- (C2);
\draw[arrow] (C2) -- (C3);
\draw[arrow] (C3) -- (C4);

\draw[arrow] (S1) -- (S2);
\draw[arrow] (S2) -- (S3);
\draw[arrow] (S3) -- (S4);

\end{tikzpicture}
\end{center}

\begin{remark}
    Unlike the GDS channel setting, for general GDS states with different block dimensions on Bob (e.g.\ $B=\mathbb{C}^{d_0'}\oplus\mathbb{C}^{d_1'}$ with $(d_0',d_1')\neq(d_0,d_1)$), the spin alignment conjecture still suggests that one can restrict to aligned rank-one structures inside each block. However, after this restriction the different measurement branches need not be related by local isometries on $AB$, so the shortcut used in~\eqref{eq:low-bound} (identifying all branches with a single tensor-power degradable state) can fail. This makes proving additivity of one-way distillable entanglement for arbitrary generalised direct sum states comparatively harder than the channel setting.
\end{remark}

\section{Conclusion and Outlook}
We identified new structural conditions under which one-way distillable entanglement admits a single-letter formula. In particular, we introduced the notions of \textit{less noisy at level $n$} and \textit{informationally degradable} states, and showed that both classes are single-letter for $D_\to$. Informationally degradable states are moreover additive under tensor products, i.e., they exhibit strong additivity. These results demonstrate that single-letter behavior extends strictly beyond the degradable regime.

We further proved a general single-letter theorem for mixtures with orthogonal support on Alice’s side, introducing non-degradable states with “useless” components. This yields explicit examples where the distillable entanglement is simply the convex combination of the components, despite the absence of degradability.

A central conceptual ingredient is the spin-alignment phenomenon, which explains why entropy minimization for certain generalized direct-sum channels is achieved by locally aligned inputs. We established this effect for the single-copy von Neumann entropy and for arbitrary block length in the Rényi-2 case. Applying this principle, we showed that a class of generalized direct sum channels has single-letter quantum capacity, and that some particular case of the corresponding states have single-letter one-way distillable entanglement. However, we have found that even though numerical evidences for additivity of $D^{(1)}_{\to}$ are present for certain states, proving additivity is difficult.

Our results reveal new mechanisms enforcing additivity and single-letter formulas beyond degradability. An important open problem is to prove the spin alignment conjecture with von-Neumann entropy for arbitrary tensor powers. Understanding the comparative additivity of $D^{(1)}_{\to}$ and $\mathcal{Q}^{(1)}$ remains to be explored. Finally, it remains to be understood whether the techniques developed here extend to two-way distillation or to other quantum resource theories.
\bibliographystyle{marcotomPB} 
\bibliography{additivity}

\appendix

\subsection{Explicit upper bound on each alignment term for R\'enyi-2 entropy}\label{app:upper bound}
In this section, we present the details of the proof of \eqref{eqn:key upper bound Renyi}. That is 
\begin{align*}
    \Tr\big[\mathcal N_{\vec x}(\rho_{\vec x})\,\mathcal N_{\vec y}(\rho_{\vec y})\big] \le \alpha_1^{N_{00}(\vec x,\vec y)}\alpha_0^{N_{11}(\vec x,\vec y)}
(\lambda_0\lambda_1)^{N_{01}(\vec x,\vec y)+N_{10}(\vec x,\vec y)},\quad \forall \vec x,\vec y\in\{0,1\}^n,
\end{align*}
where $\mathcal N_0(\rho)=\rho\otimes\sigma_1$ and $\mathcal N_1(\rho)=\sigma_0\otimes\rho$ and the constants are given by
\begin{equation}
    \lambda_0:=\lambda_{\max}(\sigma_0),\qquad \lambda_1:=\lambda_{\max}(\sigma_1),\qquad
\alpha_0:=\Tr(\sigma_0^2),\qquad \alpha_1:=\Tr(\sigma_1^2).
\end{equation}
For $\vec x,\vec y\in\{0,1\}^n$, the counts are given by
\begin{equation}
    N_{ab}(\vec x,\vec y):=\bigl|\{i\in[n] : (x_i,y_i)=(a,b)\}\bigr|,\qquad a,b\in\{0,1\}.
\end{equation}

\begin{lemma}[Single-site adjoint compositions]\label{lem:single-site}
With respect to the Hilbert--Schmidt inner product, the adjoints satisfy
\[
\mathcal N_0^\dagger(X)=\Tr_{B_1}\!\big[(I\otimes\sigma_1)X\big],
\qquad
\mathcal N_1^\dagger(X)=\Tr_{B_0}\!\big[(\sigma_0\otimes I)X\big].
\]
Moreover, for any positive operator $X$ of compatible dimension,
\begin{align*}
\mathcal N_0^\dagger\mathcal N_0(X)&=\alpha_1\,X,\\
\mathcal N_1^\dagger\mathcal N_1(X)&=\alpha_0\,X,\\
\mathcal N_0^\dagger\mathcal N_1(X)&=\Tr(\sigma_1 X)\,\sigma_0,\\
\mathcal N_1^\dagger\mathcal N_0(X)&=\Tr(\sigma_0 X)\,\sigma_1.
\end{align*}
\end{lemma}

\begin{proof}[Proof of Lemma~\ref{lem:single-site}]
The adjoint formulas follow directly from
\[
\Tr\big[(\rho\otimes\sigma_1)X\big]=\Tr\!\big[\rho\,\Tr_{B_1}((I\otimes\sigma_1)X)\big],
\qquad
\Tr\big[(\sigma_0\otimes\rho)X\big]=\Tr\!\big[\rho\,\Tr_{B_0}((\sigma_0\otimes I)X)\big].
\]
The composition identities are then obtained by substitution. For example,
\[
\mathcal N_0^\dagger\mathcal N_1(X)
=\mathcal N_0^\dagger(\sigma_0\otimes X)
=\Tr_{B_1}\big[(I\otimes\sigma_1)(\sigma_0\otimes X)\big]
=\Tr(\sigma_1 X)\,\sigma_0,
\]
and the remaining cases are analogous.
\end{proof}
\begin{proof}[Proof of \eqref{eqn:key upper bound Renyi}] 
Fix $\vec x,\vec y$. Using adjoints and Hölder's inequality,
\begin{align}
\Tr\big[\mathcal N_{\vec x}(\rho_{\vec x})\,\mathcal N_{\vec y}(\rho_{\vec y})\big]=
\Tr\big[\rho_{\vec x}\,(\mathcal N_{\vec x}^\dagger\mathcal N_{\vec y})(\rho_{\vec y})\big]\notag\le \|\rho_{\vec x}\|_1\,\big\|(\mathcal N_{\vec x}^\dagger\mathcal N_{\vec y})(\rho_{\vec y})\big\|_\infty
= \big\|(\mathcal N_{\vec x}^\dagger\mathcal N_{\vec y})(\rho_{\vec y})\big\|_\infty. \label{eq:holder}
\end{align}
Since $\mathcal N_{\vec x}=\bigotimes_i \mathcal N_{x_i}$, we have
\[
\mathcal N_{\vec x}^\dagger\mathcal N_{\vec y}
=
\bigotimes_{i=1}^n (\mathcal N_{x_i}^\dagger\mathcal N_{y_i}).
\]
By Lemma~\ref{lem:single-site}, each tensor factor is one of four completely positive maps:
\[
\mathcal N_0^\dagger\mathcal N_0(X)=\alpha_1 X,\quad
\mathcal N_1^\dagger\mathcal N_1(X)=\alpha_0 X,\quad
\mathcal N_0^\dagger\mathcal N_1(X)=\Tr(\sigma_1 X)\sigma_0,\quad
\mathcal N_1^\dagger\mathcal N_0(X)=\Tr(\sigma_0 X)\sigma_1.
\]
Let $D:=\{i: x_i\neq y_i\}$ and $E:=\{i:x_i=y_i\}$. The equal sites contribute the scalar factor
$\alpha_1^{N_{00}(\vec x,\vec y)}\alpha_0^{N_{11}(\vec x,\vec y)}$.
On the mismatch sites $D$, the map is ``measure-and-prepare'':
it applies a weighted partial trace against $\sigma_1$ (resp.\ $\sigma_0$) and outputs $\sigma_0$ (resp.\ $\sigma_1$).
Concretely, there exist positive operators
\[
M_{\vec x,\vec y}
=
\sigma_1^{\otimes N_{01}(\vec x,\vec y)}\otimes \sigma_0^{\otimes N_{10}(\vec x,\vec y)},
\qquad
P_{\vec x,\vec y}
=
\sigma_0^{\otimes N_{01}(\vec x,\vec y)}\otimes \sigma_1^{\otimes N_{10}(\vec x,\vec y)},
\]
such that for every positive operator $X$ on the input space of $\mathcal N_{\vec y}$,
\begin{equation}\label{eq:explicit-form}
(\mathcal N_{\vec x}^\dagger\mathcal N_{\vec y})(X)
=
\alpha_1^{N_{00}}\alpha_0^{N_{11}}
\big(\Tr_D\big[(M_{\vec x,\vec y}\otimes I_E)X\big]\big)\otimes P_{\vec x,\vec y},
\end{equation}
where $\Tr_D$ is the partial trace over the mismatch input registers and $N_{ab}=N_{ab}(\vec x,\vec y)$.

Now specialize to $X=\rho_{\vec y}$ (a density operator). The operator
\[
R_E:=\Tr_D\big[(M_{\vec x,\vec y}\otimes I_E)\rho_{\vec y}\big]
\]
is positive, and hence $\|R_E\|_\infty\le \Tr(R_E)$. Moreover,
\[
\Tr(R_E)=\Tr\big[(M_{\vec x,\vec y}\otimes I_E)\rho_{\vec y}\big]\le \|M_{\vec x,\vec y}\|_\infty
=\lambda_1^{N_{01}}\lambda_0^{N_{10}},
\]
since $\Tr(K\rho)\le \|K\|_\infty\Tr(\rho)$ for $K\ge 0$ and $\Tr(\rho_{\vec y})=1$.
Also,
\[
\|P_{\vec x,\vec y}\|_\infty=\lambda_0^{N_{01}}\lambda_1^{N_{10}}.
\]
Combining with \eqref{eq:explicit-form} and using $\|A\otimes B\|_\infty=\|A\|_\infty\|B\|_\infty$ gives
\begin{equation}\label{eq:pairwise-bound}
\big\|(\mathcal N_{\vec x}^\dagger\mathcal N_{\vec y})(\rho_{\vec y})\big\|_\infty
\le
\alpha_1^{N_{00}(\vec x,\vec y)}\alpha_0^{N_{11}(\vec x,\vec y)}
(\lambda_0\lambda_1)^{N_{01}(\vec x,\vec y)+N_{10}(\vec x,\vec y)}.
\end{equation}
\end{proof}
\subsection{Generalized direct sum completely positive maps and states}
\begin{definition}[Generalized direct sum CP maps]
    Let us consider Hilbert spaces, $A\cong A_0\oplus A_1$ and $B\cong B_0\oplus B_1$. Moreover consider two CP maps, $\Phi_i: \mc{B}(A_i)\to \mc{B}(B_i)$, for $i=0,1$, with particular Kraus representations, $\{F^{(i)}_k\}$. A generalised direct sum (GDS) completely positive (CP) map, $\Phi:\mc{B}(A)\to \mc{B}(B)$ is defined by its Kraus representation: $\{F_k:= F^{(1)}_{k}\oplus F_k^{(2)}\}$.
\end{definition}
Note that in the definition, we add zero operators in either of the set of Kraus operators to make them equally long. Furthermore, different ordering of the Kraus operators in the direct sum, can yield different generalised direct sum CP maps. As these maps are not necessarily trace-preserving, in general $\sum_{k} F^{\dagger}_k F_k \neq \mb{I}_{A}$. The Choi-Jamio\l{}kowski isomorphism \cite{CHOI1975285,JAMIOLKOWSKI1972275} implies that corresponding to the completely positive map, $\Phi:\mc{B}(A)\to \mc{B}(B)$, there exists a positive semidefinite operator on the bipartite system $AB$, i.e., $J_{\Phi} \in \mathrm{Pos}(A\otimes B)$. In fact, we can write \cite{Watrous_2018}
\begin{align}
    J_{\Phi} = \sum_{k} \mathrm{vec}({F^T_k})\mathrm{vec}({F^T_k})^{\dagger}.
\end{align}
For non-zero maps, $\Tr{J_{\Phi}} \neq 0$. Therefore, we can define a corresponding GDS state $\rho^{\Phi} \in \mc{D}(A\otimes B)$ as
\begin{align}
    \rho^{\Phi} = \frac{1}{\Tr{J_\phi}}J_{\Phi}.
\end{align}
\begin{theorem}
    Let $A,B \cong \mb{C}^{d_0+d_1}$. The following bipartite state
    \begin{align}
    \rho_{AB} = \frac{1}{2d_0d_1}\sum_{i=0}^{d_0-1}\sum_{j=0}^{d_1-1} \left(\ket{i}\ket{j} + \ket{j+d_0}\ket{i+d_1}\right) \left(\bra{i}\bra{j} + \bra{j+d_0}\bra{i+d_1}\right) 
\end{align}
has single-letter one-way distillable entanglement, i.e.,
\begin{align}
    D_{\to}(\rho_{AB}) = D^{(1)}_{\to}(\rho_{AB}).
\end{align}
\end{theorem}
\begin{proof}
    First we define the instrument $\mc{T}:A\to AM$, with Kraus operators $K_{ij} = \frac{1}{\sqrt{d_1}}\ketbra{i}{i} + \frac{1}{\sqrt{d_0}}\ketbra{d_0+j}{d_0+j}$ where $i = 0,\dots, d_0-1$ and $j=0,\dots,d_1-1$. We have a total of $d_0d_1$ number of Kraus operators. In the following, we prove that $\mc{T}^{\otimes n}: A^{\otimes n} \to A^{\otimes n} M $ is the instrument that achieves the optimum for $D^{(1)}_{\to}(\rho^{\otimes n}_{AB})$ for all $n\in \mb{N}$.

\textit{\textbf{Step 1} (Block diagonal POVM elements are sufficient):}
Consider any instrument $\mathcal{T}_1:A^{\otimes n} \to A^{\otimes n} M$ with POVM elements
\begin{align*}
    E_{m,n} = \sum_{\Vec{x},\Vec{y}\in \{0,\dots,d_0+d_1-1\}^n} \;e^{m}_{\Vec{x},\Vec{y}}\;\ketbra{\Vec{x}}{\Vec{y}}.
\end{align*}
For all $\vec{x} \in  \{0,\dots,d_0+d_1-1\}^n$, define $\vec{s}(\vec{x})\in \{0,1\}^n$ such that
\begin{align}
    s_i = \begin{cases}
        0, \quad x_i\in \{0,\dots,d_1-1\}\\
        1, \quad x_i \in \{d_1,\dots,d_1+d_0-1\}.
    \end{cases}
\end{align}
The $(\vec{x},\vec{x})$-th diagonal element of $\rho^{(m)}_{B^{\otimes n}}$ is proportional to 
\begin{align}
     \Tr \Pi_{\vec{s}(\vec{x})} E_{m,n} \Pi_{\vec{s}(\vec{x})}. 
\end{align}
Here $\Pi_{\vec{x}} = \Pi_{x_1}\otimes \Pi_{x_2}\otimes\dots \otimes \Pi_{x_n}$ with $\Pi_{0} = \sum_{j=0}^{d_0-1} \ketbra{j}{j}$ and $\Pi_{1} = \sum_{j=0}^{d_1-1} \ketbra{j}{j}$. 

The state of $E^{\otimes n}$ is proportional to

\begin{align}
    \sum_{\vec{x}\in\{0,1\}^n} \bigotimes_{i=1}^n \mathcal{N}_{x_i}(\mathscr{E}_{\vec{x}}), 
\end{align}
where 
\begin{align}
    \mathscr{E}_{\vec{x}} = \Pi_{\vec{x}} E^{T}_{m,n} \Pi_{\vec{x}}.
\end{align}
Here we treat $\mathscr{E}_{\vec{x}}$ to be a positive operator on the space $\bigotimes_i\mathbb{C}^{d_{x_i}}$. Furthermore, we define, $\mathcal{N}_0:B(\mathbb{C}^{d_0})\to B(\mathbb{C}^{d_1}\otimes \mathbb{C}^{d_0})$ and $\mathcal{N}_1:B(\mathbb{C}^{d_1})\to B(\mathbb{C}^{d_1}\otimes \mathbb{C}^{d_0})$ as:
\begin{align}
    \mathcal{N}_0(\rho) = \mathbb{I}_{d_1}\otimes\rho ; \qquad \mathcal{N}_1(\sigma) = \sigma \otimes \mathbb{I}_{d_0}.
\end{align}
Note that the state of the environment does not depend on the the \textit{off-diagonal blocks} of the POVM elements, i.e., on $\Pi_{\vec{x}}E_{m,n}\Pi_{\vec{y}}$ for $\vec{x}\neq \vec{y}$. As these blocks appear only in the off-diagonal entries of $B^{\otimes n}$, we can choose POVM elements $E_{m,n}$ such that $$\Pi_{\vec{x}}E_{m,n}\Pi_{\vec{y}} = \delta_{\vec{x},\vec{y}}\Pi_{\vec{x}}E_{m,n}\Pi_{\vec{x}}$$
without decreasing the coherent information.

\textit{\textbf{Step 2} (Spin-alignment phenomena):} Furthermore, the spin alignment conjecture \ref{conj:spin-alignment} implies that the entropy of the environment is minimum when replace the POVM elements $E_{m,n}$ by positive elements
\begin{align}
    E'_{m,n} = \sum_{\vec{x}\in\{0,1\}^{\otimes n}} \Tr (\Pi_{\vec{x}} E_{m,n} \Pi_{\vec{x}})\; \bigotimes_{i=1}^n \ketbra{t(x_i)}{t(x_i)},
\end{align}
where $t(0)=0$ and $t(1)=d_0$. Let us denote the corresponding trace non-preserving instrument (which we have introduced in the definition \ref{def: Non-TP instrument}), by $\mc{T}_2: A^{\otimes n}\to A^{\otimes n} M$. Therefore, so far we have proven that for any trace-preserving instrument, $\mathcal{T}_1$, we can find a (possibly trace non-preserving) instrument, $\mathcal{T}_2$ such that 
\begin{align}
    I(A^{\otimes n}\rangle B^{\otimes n}M)_{\mc{T}_2(\rho^{\otimes n})} \geq I(A^{\otimes n}\rangle B^{\otimes n}M)_{\mc{T}_1(\rho^{\otimes n})}.
\end{align}
\textit{\textbf{Step 3} (Matching upper and lower bound on $D^{(1)}_{\to}(\rho^{\otimes n})$)}: Now, we consider the state
\begin{align}
    \sigma = \sigma^{\otimes n}_{0} = (d_0d_1)^n\; ((K_{00} \otimes \mb{I}_B) \rho_{AB} (K_{00}^\dagger \otimes \mb{I}_B))^{\otimes n},
\end{align}
and the instrument $\mathcal{T}_3:A^{\otimes n}\to A^{\otimes n} M$ given by the POVM elements: 
\begin{align*}
    F_m = \sum_{\vec{x}\in\{0,1\}^n} \frac{1}{d_0^{N_0(\vec{x})}}\frac{1}{d_1^{N_1(\vec{x})}}\Tr \Pi_{\vec{x}} E_{m,n} \Pi_{\vec{x}}  \bigotimes_{i=1}^n \ketbra{t(x_i)}{t(x_i)};\qquad F = \mb{I}_{A^{\otimes n}} - \sum_{m}F_m.
\end{align*}
It immediately follows that 
\begin{align}
    I(A^{\otimes n}\rangle B^{\otimes n}M)_{\mc{T}_2(\rho^{\otimes n})} = I(A^{\otimes n}\rangle B^{\otimes n}M)_{\mc{T}_3(\sigma)}.
\end{align}

Now $\sigma_0$ is a degradable state, as can be seen as follows:
\begin{align}
    \sigma^{AB}_0 = \frac{1}{2}\left(\frac{1}{\sqrt{d_1}}\ket{0,0}+\frac{1}{\sqrt{d_0}}\ket{d_0,d_1}\right)\left(\frac{1}{\sqrt{d_1}}\bra{0,0}+\frac{1}{\sqrt{d_0}}\bra{d_0,d_1}\right) + \frac{1}{2d_1}\sum_{j=1}^{d_1-1} \ketbra{0,j}{0,j} + \frac{1}{2d_0}\sum_{i=1}^{d_0-1} \ketbra{d_0,i+d_1}{d_0,i+d_1}
\end{align}
For any purification $\ket{\sigma_0}_{ABE}$, by applying a fully amplitude damping channel from $\ket{d_1}\to \ket{0}$ on $B$ alone, we obtain a state on $AB$ that is equivalent to the state $\sigma^{AE}_0$, upto local isometries on $B$. Therefore, $\sigma = \sigma_0^{\otimes n}$ is degradable as well.

We know that for degradable states, the trivial instrument achieves the optimal coherent information. We get the following chain of inequality,
\begin{align}
    I(A^{\otimes n}\rangle B^{\otimes n})_\sigma \geq I(A^{\otimes n}\rangle B^{\otimes n}M)_{\mc{T}_3(\sigma)} =I(A^{\otimes n}\rangle B^{\otimes n}M)_{\mc{T}_2(\rho^{\otimes n})} \geq I(A^{\otimes n}\rangle B^{\otimes n}M)_{\mc{T}_1(\rho^{\otimes n})}
\end{align}
As $\mc{T}_1$ is any trace-preserving instrument, this inequality implies that 
\begin{align}\label{eq:up-b}
    I(A^{\otimes n}\rangle B^{\otimes n})_\sigma \geq D^{(1)}_{\to}(\rho^{\otimes n}_{AB})
\end{align}
Now we note that for all $x_i\in (0,\dots, d_0-1)$ and $y_i \in (0,\dots,d_1-1)$, the states $$(d_0d_1)^n\; \bigotimes_{i=1}^n(K_{x_iy_i} \otimes \mb{I}_B) \rho_{AB} (K_{x_iy_i}^\dagger \otimes \mb{I}_B)$$
are related to $\sigma$ via local unitaries on $A^{\otimes n}$ and $B^{\otimes n}$, and there are $(d_0d_1)^n$ numbers of such states. Therefore, they all have the same coherent information as $I(A^{\otimes n}\rangle B^{\otimes n})_\sigma$. Noting that $\bigotimes^n_{i=1} K_{x_iy_i}$ are nothing but the Kraus operators of the proposed optimum instrument $\mc{T}^{\otimes n}$, we have 
\begin{align}\label{eq:low-b}
    I(A^{\otimes n}\rangle B^{\otimes n})_\sigma =I(A^{\otimes n}\rangle B^{\otimes n}M)_{\mc{T}^{\otimes n}(\rho^{\otimes n})}\leq D^{(1)}_{\to}(\rho^{\otimes n}_{AB})
\end{align}
where the last inequality follows from the definition of $D^{(1)}_{\to}(\rho^{\otimes n}_{AB})$.

\textit{\textbf{Step 4} (Additivity of $D^{(1)}_{\to}(\rho^{\otimes n}_{AB}$)
)}: From \ref{eq:up-b} and \ref{eq:low-b}, we get
\begin{align}
    D^{(1)}_{\to}(\rho^{\otimes n}_{AB}) &= I(A^{\otimes n}\rangle B^{\otimes n})_\sigma\nonumber= I(A^{\otimes n}\rangle B^{\otimes n})_{\sigma^{\otimes n}_{0}} \nonumber= nI(A\rangle B)_{\sigma_0} = n D^{(1)}(\rho_{AB}).
\end{align}
Therefore, we finally have
\begin{align}
    D_{\to} (\rho_{AB}) = \lim_{n\to \infty} \frac{1}{n}D^{(1)}_{\to}(\rho^{\otimes n}_{AB}) = D^{(1)}(\rho_{AB}).
\end{align}
\end{proof}
\newpage

\newpage

\end{document}